\documentclass[11pt,letterpaper]{article}
\usepackage{standalone}
\usepackage{etoolbox}
\newbool{arxiv}
\setbool{arxiv}{false}

\usepackage[letterpaper,margin=1in]{geometry}
\usepackage{amsthm}
\usepackage{amsmath}
\usepackage{amssymb}
\usepackage{thmtools}
\usepackage{thm-restate}
\usepackage{times}
\usepackage{mathtools}

\usepackage{stmaryrd}%fancy brackets \llbracket  1 \rrbracket     \quad \llparenthesis 2 \rrparenthesis  

\usepackage{placeins}

\usepackage[table]{xcolor}

\definecolor{linkcol}{rgb}{0,0,0.38}
\definecolor{citecol}{rgb}{0.8,0,0}
\definecolor{urlcol}{rgb}{0.1,0.35,0}

\usepackage[pdfpagelabels,bookmarks=false,hyperfootnotes=false]{hyperref}
\hypersetup{colorlinks, linkcolor=linkcol, citecolor=citecol, urlcolor=urlcol}

\usepackage[utf8]{inputenc}

\usepackage[
backend=biber,
style=alphabetic,
citestyle=alphabetic,
maxalphanames=4,
maxcitenames=99,
mincitenames=98,
maxbibnames=99,
giveninits=true,
]{biblatex}
\usepackage[nameinlink]{cleveref}

\usepackage{lmodern}

\usepackage{bbm}
\usepackage{thm-restate}

\newtheoremstyle{light} %
    {\topsep}                    %
    {\topsep}                    %
    {\itshape}                   %
    {}                           %
    {\scshape}                   %
    {.}                          %
    {.5em}                       %
    {}  %

\newtheorem{theorem}{Theorem}[section]
\newtheorem{lemma}[theorem]{Lemma}

\newtheorem{proposition}[theorem]{Proposition}

\newtheorem{claim}[theorem]{Claim}

\theoremstyle{light}

\makeatletter
\if@cref@capitalise
\crefname{claiminproof}{Claim}{Claims}
\else
\crefname{claiminproof}{claim}{claims}
\fi

\if@cref@capitalise
\crefname{algocf}{Algorithm}{Algorithms}
\else
\crefname{algocf}{algorithm}{algorithms}
\fi

\if@cref@capitalise
\crefname{conjecture}{Conjecture}{Conjectures}
\else
\crefname{conjecture}{conjecture}{conjectures}
\fi

\if@cref@capitalise
\crefname{thm}{Theorem}{Theorems}
\else
\crefname{thm}{theorem}{theorems}
\fi

\if@cref@capitalise
\crefname{lem}{Lemma}{Lemmas}
\else
\crefname{lem}{lemma}{lemmas}
\fi

\makeatother

\usepackage{url}
\usepackage{array}

\usepackage{setspace}

\usepackage{wrapfig}

\makeatletter
\newcommand{\labeltarget}[1]{\Hy@raisedlink{\hypertarget{#1}{}}}
\makeatother

\usepackage{upgreek}

\usepackage{complexity}

\usepackage[inline]{enumitem}
\usepackage{moreenum}

\usepackage[super]{nth}

\usepackage{bm}

\usepackage{xspace}

\usepackage{ifthen}

\usepackage[immediate]{silence}
\usepackage{sectsty}
\DeactivateWarningFilters[tmp]
\allsectionsfont{\boldmath}

\setlist[enumerate]{nosep,topsep=0.1em}
\setlist[enumerate,1]{label=(\roman*), leftmargin=2.2em}
\setlist[itemize]{nosep,topsep=0.3em}

\usepackage[textsize=footnotesize, color=blue!30!white]{todonotes}
\ifbool{arxiv}{
}{

}

\usepackage{xfrac}

\usepackage{tikz}
\usetikzlibrary{calc}
\usetikzlibrary{math}
\usetikzlibrary{shapes.geometric}
\usetikzlibrary{arrows}
\usetikzlibrary{decorations.pathreplacing, decorations.pathmorphing}

\usepackage{tcolorbox}
\tcbuselibrary{skins}

\usepackage{float}
\usepackage{graphicx}
\graphicspath{{graphics/}}
\makeatletter
\newcommand\appendtographicspath[1]{%
  \g@addto@macro\Ginput@path{#1}%
}
\makeatother

\usepackage[margin=10pt,font=small,labelfont=bf,skip=-5pt]{caption}
\usepackage{subcaption}

\usepackage[vlined,ruled,algo2e]{algorithm2e}
\SetAlgoSkip{bigskip}
\SetAlgoInsideSkip{smallskip}

\usepackage{mdframed}

\let\truehypersetup\hypersetup
\renewcommand\hypersetup[1]{}
\usepackage{bigfoot}
\let\hypersetup\truehypersetup
\definecolor{morado}{rgb}{1,0,1}
\definecolor{verde}{rgb}{0,0.8,0}
\definecolor{amarillo}{rgb}{1,0.6,0.2}
\definecolor{azul}{rgb}{0,0.6,1}
\definecolor{gris}{rgb}{0.12549019607843137,0.12549019607843137,0.12549019607843137}

\newcommand{\junk}[1]{}

\newcommand{\bon}{\ensuremath{\mathbbm{1}}}
\renewcommand{\R}{\ensuremath{\mathbb{R}}}%real numbers
\newcommand{\Dmax}{\ensuremath{d_{{\max}}}}%maximum demand

\newcommand{\F}{\ensuremath{\mathcal{F}}}% faces of G
\newcommand{\fout}{\ensuremath{f_{\infty}}}% outer face of G
\renewcommand{\W}{\ensuremath{\mathcal{W}}}%partition set for outerplanar theorem
\newcommand{\Wk}{\ensuremath{\mathcal{W}^{(\ite)}}}%partition set for outerplanar theorem
\newcommand{\Wkplus}{\ensuremath{\mathcal{W}^{(\ite+1)}}}%partition set for outerplanar theorem
\newcommand{\Wkminus}{\ensuremath{\mathcal{W}^{(\ite-1)}}}%partition set for outerplanar theorem

\newcommand{\I}{\ensuremath{\mathcal{I}}}%multicommodity flow instance
\renewcommand{\RL}{\ensuremath{\mathcal{RL}}}%ring loading instance
\newcommand{\Alg}{\ensuremath{\mathcal{A}_{RL}}}%ring loading instance

\newcommand{\openleft}{\ensuremath{\llparenthesis}}% 
\newcommand{\openright}{\ensuremath{\rrparenthesis}}%
\newcommand{\closedleft}{\ensuremath{[}}%
\newcommand{\closedright}{\ensuremath{]}}%

\newcommand{\GD}{\ensuremath{G^D}}%dual graph
\newcommand{\VD}{\ensuremath{V^D}}%dual vertices
\newcommand{\ED}{\ensuremath{E^D}}%dual edges

\newcommand{\cp}{\ensuremath{c^{(\ite)}}}%capacities during iterations of algorithm
\newcommand{\crl}{\ensuremath{c_{\text{rl}}}}%ring loading instance capacities
\newcommand{\lab}{\ensuremath{\textsc{label}}}%edge label
\newcommand{\labk}{\ensuremath{\textsc{label}^{\hspace{-0.07cm}(k)}}}%edge label
\newcommand{\labone}{\ensuremath{\textsc{label}^{\hspace{-0.07cm}(1)}}}%edge label
\newcommand{\labkminus}{\ensuremath{\textsc{label}^{\hspace{-0.07cm}(k-1)}}}%

\newcommand{\Dp}{\ensuremath{D^{(\ite)}}}%New demand edges
\newcommand{\Dpminus}{\ensuremath{D^{(\ite-1)}}}%New demand edges
\newcommand{\Dpplus}{\ensuremath{D^{(\ite+1)}}}%New demand edges
\newcommand{\dpr}{\ensuremath{d^{(\ite)}}}%New demands
\newcommand{\ite}{\ensuremath{k}}%iteration number
\makeatletter
\let\@@pmod\pmod
\DeclareRobustCommand{\pmod}{\@ifstar\@pmods\@@pmod}
\def\@pmods#1{\mkern8mu({\operator@font mod}\mkern 6mu#1)}
\makeatother

\makeatletter
\let\@@mod\mod
\DeclareRobustCommand{\mod}{\@ifstar\@mods\@@mod}
\def\@mods#1{\mkern8mu{\operator@font mod}\mkern 6mu#1}
\makeatother

\definecolor{green}{rgb}{0.4,0.85,0.6}
 
\makeatletter
\def\@fnsymbol#1{\ensuremath{\ifcase#1\or *\or %
\ddagger\or
    \mathsection\or \mathparagraph\or \|\or **\or \dagger\dagger
    \or \ddagger\ddagger \else\@ctrerr\fi}}
\makeatother

\title{Unsplittable Multicommodity Flows in Outerplanar Graphs%
\ifbool{arxiv}{}{%
\thanks{This work was partially supported by the NSERC Discovery Grant Program, grant number RGPIN-2024-04532.}}
}

\ifbool{arxiv}{
\author{}
}{
\author{
David Alem\'an-Espinosa\thanks{
Department of Combinatorics and Optimization, University of Waterloo, Waterloo, Canada.
Email: \href{mailto:dalemanespinosa@uwaterloo.ca}%
{dalemanespinosa@uwaterloo.ca}.
}
\and
Nikhil Kumar\thanks{
Department of Combinatorics and Optimization, University of Waterloo, Waterloo, Canada.
Email: \href{mailto:nikhil.kumar2@uwaterloo.ca}%
{nikhil.kumar2@uwaterloo.ca}.}
}
}

\date{}

\begin{document}
\maketitle
% \thispagestyle{empty}
%  \addtocounter{page}{-1}
%  \enlargethispage{-1cm}
%
\begin{abstract}
We consider the problem of multicommodity flows in outerplanar graphs. Okamura and Seymour~\cite{okamura1981multicommodity} showed that the cut-condition is sufficient for routing demands in outerplanar graphs. We consider the unsplittable version of the problem and prove that if the cut-condition is satisfied, then we can route each demand along a single path by exceeding the capacity of an edge by no more than $\frac{18}{5} \cdot \Dmax$, where $\Dmax$ is the value of the maximum demand. 
\end{abstract}

\section{Introduction}
Given a graph $G$ with edge capacities and multiple source-sink pairs, each with an associated demand, the multicommodity flow problem consists of routing all demands simultaneously without violating edge capacities. The graph obtained by including an edge $(s_{i},t_{i})$ for a demand with source-sink $s_{i},t_{i}$ is called the demand graph and will be denoted by $H$. The problem was first formulated in the context of VLSI routing in the 70's and since then it has been the subject of a long and impressive line of work.

A multicommodity flow is called unsplittable if all the flow between a source-sink pair is routed along a single path. In general, existence of a feasible (fractional) flow does not imply the existence of an unsplittable flow, even in very simple settings. 
Consider a four-cycle $v_1v_2v_3v_4v_1$ with uniform edge capacity $d\geq0$, where the demand edges are $(v_1,v_3)$ and $(v_2,v_4)$, both of value $d$.
A feasible flow is given by evenly splitting the demand across the two possible paths for each demand edge. However, any unsplittable flow would violate the capacity of at least one edge by $d$. 
The above leads to a natural question: given a feasible flow, does there exist an unsplittable flow which satisfies all the demands and violates the edge capacities (in an additive sense) by at most a small factor times the value of the maximum demand $\Dmax$? 
This problem was first considered by Kleinberg~\cite{kleinberg1996single} in the setting where all the demands are incident at a single source. He also explored connections to some well-studied combinatorial optimization problems such as the generalized assignment problem~\cite{shmoys1993approximation}.
Dinitz, Garg and Goemans~\cite{dinitz1999single} proved that in the single-source setting,
any feasible fractional flow can be converted into an unsplittable flow that violates the edge capacities by no more than $\Dmax$.
%The following is therefore another natural question: which non-trivial graph classes guarantee the existence of an unsplittable flow that incurs an additive violation of the edge capacities by at most $\mathcal{O}(\Dmax)$, for an arbitrary $H$ and whenever a feasible fractional flow exists?

The problem is significantly less understood when the demand graph $H$ is arbitrary.
Schrijver, Seymour and Winkler~\cite{schrijver1998ringloading} proved that if $G$ is a cycle, then any feasible multicommodity flow can be converted into an unsplittable one that violates the edge capacities by at most $\frac{3}{2}\cdot \Dmax$. Before our work, cycles were the only known nontrivial class of graphs for which an unsplittable flow was guaranteed to exist, incurring at most an additive \(\mathcal{O}(\Dmax)\) violation of edge capacities, whenever a feasible flow existed. Our main contribution is extending this result to the class of \emph{outerplanar graphs}.

In particular, we show that any feasible multicommodity flow in an outerplanar graph can be converted into an unsplittable one that violates the edge capacities by no more than $\frac{18}{5}\cdot \Dmax$. 
This improves a recent result of Shapley and Shmoys~\cite{shapleyshmoys}, who proved that if $G$ is an outerplanar graph, then one can convert any feasible multicommodity flow into an unsplittable one that violates the edge capacities by at most $\mathcal{O}(\log k)\cdot\Dmax$, 
where $k$ denotes the number of faces in the graph. 

To prove our result, we crucially rely on the fact that there is a nice characterization for the existence of a feasible flow in outerplanar graphs. A necessary condition for the existence of a feasible flow is that the total demand across a cut should not exceed the capacity of the cut. 
This condition is called the $\emph{cut-condition}$ and is known to be sufficient in many cases, such as: (a) when all demand edges in $H$ are incident on a single vertex (the single source setting)~\cite{ford1956maximal}, (b) when $G$ is planar and all demand edges in $H$ are incident on one face (the Okamura-Seymour setting)~\cite{okamura1981multicommodity}, and (c) when $G+H$ is planar \cite{seymour1981odd}. In outerplanar graphs, all the vertices are incident on one face. Hence, the result of Okamura-Seymour~\cite{okamura1981multicommodity} implies that the cut-condition is sufficient for a feasible routing if $G$ is outerplanar and $H$ is arbitary. In general, the cut-condition is not sufficient for the existence of a feasible flow, even for very small instances. It is worthwhile to point out that all the results mentioned above for unsplittable flows are for instances where the cut-condition is also sufficient for the existence of a feasible multicommodity flow, although none of the proofs use this fact explicitly. Our result can be stated equivalently as follows: if $G$ is outerplanar and the cut-condition is satisfied, then there exists an unsplittable flow satisfying all the demands that does not violate the capacity of any edge by more than $\frac{18}{5}\cdot \Dmax$. 

\section{Definitions and Preliminaries}

For an integer $n\geq 1$, we use $[n]$ to denote $\{1,\hdots,n\}$. For any universe $U$, function $c:U\mapsto\R$, set $S\subseteq U$ and an element $e\in S$, we use $S-e$ and $S+e$ to denote $S\setminus\{e\}$ and $S\cup\{e\}$, respectively. We use $c(S)$ to denote $\sum_{g\in S}c(g)$.
%\textcolor{red}{add further notations}
%We say that a cycle is simple if it is 2-vertex connected. We say that a path is simple if no vertex has degree more than 2. For a path $p$ or cycle $c$, we overload notation and use $p$ and $c$ to also denote the set of vertices and edges on $G$. 

\subsection{Multicommodity Flows}
An instance $\I(G,c,D,d)$ of {\em multicommodity flow} is given by an undirected graph $G=(V,E)$ with edge capacities $c:E\mapsto\R_{\geq 0}$, and a graph $H=(V,D)$ with demands $d:D\mapsto \R_{\geq 0}$. 
%We note that $H$ could have parallel edges but no loops. We assume that $G$ has no parallel edges or loops.
We note that $G$ and $H$ could have parallel edges but no loops.
We refer to $G$ and $E$ as the supply graph and supply edges, respectively. We refer to $H$ and $D$ as the demand graph and demand edges, respectively. 
%We denote the instance by $\I(G,H)$ if the edge capacities and demands are clear from the context.
For a demand edge $g \in D$, we denote by $\mathcal{P}_{g}$ the set of all (simple) paths between endpoints of $g$ in $G$. For a $g \in D$, a demand flow $x_g$ is an assignment of non-negative real numbers to paths in $\mathcal{P}_g$, i.e.~$x_g:\mathcal{P}_g \mapsto \R_{\geq 0}$. A collection of demand flows $x=\{x_g~|~g \in D\}$ constitutes a flow for the instance. For $D' \subseteq D$, we use $\sum_{g \in D'} x_g$ to denote the collection of demand flows in $D'$. 
%\textcolor{red}{change f to x}

For an edge $e \in E$ and demand $g \in D$, we use $x_g(e)$ to denote the total flow going through an edge $e$ in $x_g$, i.e.~$x_g(e)=\sum_{p: e \in p} x_g(p)$. For an edge $e \in E$ and flow $x$, we use $x(e)$ to denote the total flow going through an edge $e$ in $x$, i.e.~$x(e)=\sum_{g \in D} x_g(e)$. We say that a flow $x$ is \emph{feasible} if it does not violate the edge capacity constraints, i.e.~$x(e) \leq c(e)$ for all $e \in E$, and satisfies all the demands, i.e.~$\sum_{p \in \mathcal{P}_g} x_g(p) = d(g)$ for all $g \in D$. 
%\david{There are parallel demands (at least in the instance we create for outerplanar graphs). Can you please address this in the definition below.}

We say that a flow $x$ is \emph{unsplittable} if at most one of the paths in $\mathcal{P}_{g}$ is assigned a non-zero value in $x_g$. We use $\Dmax:=\max_{g\in D} d(g)$ to denote the maximum demand value. We say that that an unsplittable flow $x$ is {\em $\alpha$-feasible} if $x(e) \leq c(e) + \alpha \cdot \Dmax$ for all $e \in E$ and $\sum_{p \in \mathcal{P}_g} x_g(p) = d(g)$ for all $g \in D$.

\subsection{The cut-condition}

For $S \subseteq V$, a cut $(S,V \setminus S)$ is a bi-partition of the vertex set. We refer to the cut $(S,V \setminus S)$ by $S$ for convenience. Given a set of edges $X \subseteq E$, we will use $c(X)$ to denote the sum of capacities of edges in $X$. Given a set of edges $Y \subseteq D$, we will use $d(Y)$ to denote the sum of demands edges in $Y$. The set of supply edges going across the cut $S$ is denoted by $\delta_{G}(S)$. Similarly, $\delta_{H}(S)$ denotes the set of demand edges going from $S$ to $V \setminus S$. One necessary condition for the existence of a feasible flow is as follows: $c(\delta_{G}(S)) \geq d(\delta_{H}(S))$ for every $S \subseteq V$. In other words, across every cut, the total capacity of supply edges should be at least the total value of demand edges. This condition is also known as the \emph{cut-condition}. In general, the cut-condition is not sufficient for a feasible routing; see \cite{okamura1981multicommodity} for an example. 
The following classic result identifies a setting where the cut-condition is also sufficient for routing demands in planar graphs. We will be invoking this to prove our main result.
\begin{theorem} [Okamura-Seymour~\cite{okamura1981multicommodity}]\label{OS}
Let $\I(G,c,D,d)$ be an instance of multicommodity flow such that $G$ is planar and all the edges of $H$ are incident on a fixed face \footnote{i.e.~there exists a face $f$ such that for each $(u,v) \in D$, both $u$ and $v$ lie on $f$}, then the cut-condition is necessary and sufficient for the existence of a feasible flow.
\end{theorem}

We say that $g \in D$ is \emph{good} if there exists an edge $e \in E$ with exactly the same end-points as $g$. A demand is \emph{bad} if it is not good. 
%Given a \emph{good} demand with end-points $s,t$ and a vertex $u \in V$, the process of forming a new instance by replacing $g$ with two new demands $g_1,g_2$, with end-points $s,u$ and $s,v$ is called \emph{pinning}.
%\david{Changed the following sentence.}
Given a bad demand $g\in D$ with end-points $s,t$ and a vertex $u \in V$, the process of forming a new instance by replacing $g$ with two new demands $g_1,g_2$, with end-points $s,u$ and $u,t$ is called \emph{pinning}.
Note that $g_1,g_2$ above have the same demand value as $g$. Given an instance where the cut-condition is satisfied, we say that a given pinning operation is \emph{feasible} if the resulting instance also satisfies the cut-condition. 

\subsection{Planar Graphs}\label{section:planar}
We assume a fixed planar embedding of $G$. 
Without loss of generality, one can assume that $G$ is 2-vertex connected. If there is a cut-vertex $v$ and $ab$ is a demand separated by the removal of $v$, then replacing $ab$ by $av,vb$ maintains the cut-condition. By doing this for every cut-vertex and each demand separated by them, we get separate smaller instances for each 2-vertex connected component. Hence, every vertex is a part of a cycle corresponding to some face. We abuse notation and use $f$ to also denote the edges and vertices associated with the cycle of face $f$. Given a set $S\subseteq V$, we denote the subgraph induced by vertices in $S$ as $G[S]$. We call a subset $A \subseteq V$ \emph{central} if both $G[A]$ and $G[V-A]$ are connected. A cut $(A, V \setminus A) $ is called central if $A$ is central. The following is well-known.
\begin{lemma}[\cite{schrijver2003combinatorial}]\label{central}
$(G,H)$ satisfies the cut-condition if and only if $c(\delta_G(A)) \geq d(\delta_H(A))$ for all central sets $ A \subseteq V$.
\end{lemma}
We use $\F$ to denote the set of all faces of $G$. The $\bf{dual}$ of a planar graph, denoted by $\GD=(\VD,\ED)$, is defined as follows: 
%$\VD= \F$ and if $f_{i},f_{j} \in \F$ share an edge in $G$, then $(f_{i},f_{j}) \in \ED $. 
The vertex set $\VD$ of $\GD$ is the face set $\F$ of $G$.
There is a one-to-one correspondence between $E$ and $\ED$,
where every edge $e\in E$ induces precisely one edge $e^D$ of $\ED$ incident with the two faces of $G$ that contain $e$. 
%Observe that $|\{e\in E:e\in f_1\cap f_2\}|=k$ for two distinct faces $f_1,f_2\in\F$, if and only if there are $k$ edges in $\GD$ whose endpoints are $f_1$ and $f_2$. 
The $\bf{weak}$ $\bf{dual}$ of $G$, denoted by $G^*$, is defined as the induced graph $\GD[\VD-\fout]$,
%\footnote{$\GD-\fout$ denotes the induced graph $\GD[\VD-\fout]$.}
where $\fout$ denotes the unbounded face of $G$. 
%We overload notation and also use $\fout$ to denote the set of edges on the unbounded face.
The following are well-known facts;
see for instance Section 14.3 in~\cite{godsil2001algebraicGraphTheory}.
\begin{theorem}\label{theorem:dualofplanarGcycles}
Given a planar graph $G=(V,E)$ and $F\subseteq E$, $F$ is the set of edges of a central cut in $G$ if and only if $\{e^D:e\in F\}$ is the set of edges of a simple circuit\footnote{A simple circuit is a closed walk with at least one edge, where the only repeated vertices are the first and the last.} in $\GD$.
\end{theorem}
\begin{proposition}\label{fact:weakdualtree}
Let $G$ be a 2-vertex connected outerplanar graph. Then $G^*$ is a tree.
\end{proposition}
%\vspace*{-3ex}
\subsection{The Ring Loading Problem}
When $G$ is a cycle, we refer to the unsplittable multicommodity flow instance as a {\em ring-loading} instance. Such instances were first studied by Cosares and Saniee~\cite{cosares1994optimization}. For such instances, we assume that $G$ has vertex set $V:=[n]$ and edge set $E=\{\{i,i+1\}\}_{i\in[n]}$, where we define $n+1:=1$. We use $\RL([n],c,D,d)$ to denote a ring-loading instance where $c,D,d$ are defined as above.

Theorem~\ref{OS} implies that the cut-condition is sufficient for the existence of a feasible flow if $G$ is a cycle and $H$ is arbitrary. In their seminal paper, Schrijver, Seymour and Winkler~\cite{schrijver1998ringloading} showed that if the cut-condition is satisfied for a ring-loading instance, then there exists a $\frac{3}{2}$-feasible unsplittable flow. This was later improved to $\frac{19}{14}$ by Skutella~\cite{skutella2016ringloading}, and ultimately to $\frac{13}{10}$ by D\"{a}ubel~\cite{daubel2019ringloadingBestUpperbound}. Skutella~\cite{skutella2016ringloading} also showed that it is not possible to obtain a $\alpha$- feasible unsplittable flow for $\alpha\geq \frac{11}{10}$, improving upon the previous best lower bound of $\frac{101}{100}$ provided in \cite{schrijver1998ringloading}. Let $\alpha_{\textsc{rl}}$ denote the minimum value of $\alpha$ that can be achieved by a polynomial time algorithm for the ring-loading problem. We know that $\alpha_{\textsc{rl}} \in [\frac{11}{10}, \frac{13}{10}]$, and it is an interesting open question to pin down its precise value.
%Prior to our work, cycles were the only (2-connected) class of supply graphs for which an additive constant-factor approximation algorithm was known.  
%\vspace*{-1ex}
\section{Our Contributions}

As mentioned above, there exists a feasible flow for a multicommodity instance on outerplanar graphs if and only if the cut-condition is satisfied. We study the unsplittable version of the problem, and show that if the cut-condition is satisfied for a multicommodity flow instance on an outerplanar graph, then there exists a polynomial time computable $(2 \cdot \alpha_{\textsc{rl}}+1)$-feasible unsplittable flow. By the result of D\"{a}ubel~\cite{daubel2019ringloadingBestUpperbound}, $\alpha_{\textsc{rl}} \leq \frac{13}{10}$, and hence we obtain a $\frac{18}{5}$-feasible unsplittable flow. We use the results on unsplittable flow for the ring-loading problem as a blackbox, and any improvement in the upper bound for $\alpha_{\textsc{rl}}$ would also imply an improved bound for outerplanar graphs.
\begin{restatable}{rethm}{theoremmain}\label{theorem:outerplanar}
Let $\I(G,c,D,d)$ be an instance of the multicommodity flow problem such that $G$ is outerplanar and the cut-condition is satisfied.
Then there exists a polynomial time computable $(2 \cdot \alpha_{\textsc{rl}}+1)$-feasible unsplittable flow for the instance.
\end{restatable}

We also prove the following theorem, which is useful in proving Theorem~\ref{theorem:outerplanar}. 
\begin{restatable}{rethm}{theoremgood}\label{theorem:G+Hplanar}
Let $\I(G,c,D,d)$ be an instance of the multicommodity flow problem such that $G$ is outerplanar, all the demand edges are $\emph{good}$, and the cut-condition is satisfied. Then there exists a polynomial time computable $1$-feasible unsplittable flow for the instance.
\end{restatable}

\subsection{Overview of Techniques}
%We first give an overview of the proof of Theorem~\ref{theorem:outerplanar}.
%and defer the details of proof of Theorem~\ref{theorem:G+Hplanar} to the full version.
Demand pinning is a well known technique that has been used to prove similar results. Recall that given a bad demand $g$ with end-points $s,t$ and an arbitrary vertex $u$, a feasible pinning involves creating a new instance that satisfies the cut-condition by replacing $g$ by $g_1,g_2$ with end-points $s,u$ and $u,t$. In many cases, it is possible to show that such a feasible pinning always exists and this immediately implies that the cut-condition is sufficient for feasibly routing all the demands. For example, when $G$ is an outerplanar graph such that $G+H$ is Eulerian, all the demand values are one and there are no good demand edges \footnote{if there is a good demand edge, we can always create a smaller instance by deleting that demand edge and reducing the capacity of the corresponding supply edge by one unit}, then there always exists a feasible pinning.
Unfortunately, the value of demands could be arbitrary in our setting and such a pinning is not guaranteed to exist. Fortunately, this does not rule out a feasible pinning if we are also allowed to increase the edge capacities, and we crucially exploit this fact to design our algorithm.

We deviate from the traditional pinning paradigm described above, and instead of pinning one demand in an iteration, we pin a set of demands simultaneously. As far as we know, this is the first time such an operation has been used in this setting. This allows us to combine the operation of pinning and increasing edge capacities in a careful manner, and ensure that the capacity of an edge is not increased by too much during the whole process. In particular, we show that there exists a careful sequence of pinning and edge capacity increases such that (i) the resulting intermediate instances satisfy the cut-condition (ii) the capacity of no edge is increased by more than $2 \cdot \alpha_{\textsc{RL}} \cdot d_{\max}$ during the whole process. Furthermore, each demand $h$ in the final instance is \emph{good}, i.e.,~there is an edge of $G$ with the same end-points as $h$. This immediately implies that the resulting instance satisfies the conditions of Theorem~\ref{theorem:G+Hplanar}, and that there exists a $(2 \cdot \alpha_{\textsc{RL}}+1)$-feasible unsplittable flow. 
We prove the main result in section~\ref{section:outer}, and postpone the proof of Theorem~\ref{theorem:G+Hplanar} to section~\ref{section:good}.

%\vspace*{-1ex}
\section{Unsplittable Flows in Outerplanar Graphs}\label{section:outer}
We assume that $G=(V,E)$ is a simple outerplanar graph for the remainder of this section.
We describe in section~\ref{parallele} how to handle the case in which $G$ has parallel edges.
As noted in Section~\ref{section:planar}, we work with a fixed planar embedding of $G$ and assume that $G$ is 2-vertex connected. Since $G$ is 2-vertex connected, every face of $G$ is a cycle.
We assume that the vertex set of $G$ is $V=[n]:=\{1,\hdots,n\}$. 
In order to simplify the notation, we assume that $n+1$ and $0$ denote vertices $1$ and $n$, respectively. We denote the set of faces of $G$ by $\F$, and use $\fout$ to denote the unbounded face of $G$.
We assume that $\fout=12\hdots n1$, i.e.~the nodes of the cycle $\fout$ are ordered in a clockwise manner and the edges of $\fout$ are $\{i,i+1\}$ for $i\in[n]$.

We view the nodes of every face $f=i_1i_2\hdots i_\ell i_1\in\F$ as being ordered clockwise around $f$ as well. In particular, we assume that $1\leq i_1<i_2<\hdots<i_j\leq n$, and $1\leq i_{j+1}<i_{j+2}<\hdots<i_\ell<i_1$ for some $j\in[\ell]$.
If it is clear from the context which face is being considered, we will sometimes use $i_{\ell+1}$ and $i_0$ to denote $i_1$ and $i_\ell$, respectively. 

For any $i,j\in [n]$, we use $\closedleft i,j\closedright$ to denote $\{i,i+1,\hdots,j\}$ if $i\leq j$, and $\{i,i+1,\hdots,n,1,\hdots j\}$ if $j<i$.
Analogously, for any $i\neq j\in [n]$ we use $\openleft i,j\openright$ to denote $\{i+1,i+2,\hdots,j-1\}$ if $i<j$, and $\{i+1,i+2,\hdots,n,1,\hdots,j-1\}$ if $i>j$. 
Recall that a set $S \subseteq V$ is called central if both $G[S]$ and $G[V \setminus S]$ are connected. Since $G$ is outerplanar, all central sets are of the form $[i,j]$ for $i,j\in[n]$. Observe that $[i,j]=[n]$ if $i=j+1$,
so we may assume that $i\neq j+1$ when considering central cuts.

Let $G^*$ be the weak dual of $G$. By Proposition \ref{fact:weakdualtree}, $G^*$ is a tree. 
For $f_1,f_2\in \F-\fout$, we use $G^*_{f_1,f_2}$ to denote the unique path between $f_1$ and $f_2$ in $G^*$. We overload notation, and use $G^*_{f_1,f_2}$ as a set of edges, or a set of faces; the meaning will be clear from the context. When considering any $e\in E\setminus\fout$, we also use $e$ to denote the edge corresponding to $e$ in $G^*$; it will be clear from the context which edge and graph ($G$ or $G^*$) we are considering. The following simple yet crucial lemma will prove to be extremely useful.
%We defer its proof to \cref{proofapp_i_j}.
%in is the main reason why all multicommodity flow instances considered during a run of algorithm \ref{algorithmFaceAug} satisfy the cut-condition, and are therefore fractionally feasible.  
\begin{lemma}\label{lemma:usefullemmafor[i,j]}
Let $i,j\in[n]$ with $i\neq j+1$, and $f=i_1i_2\hdots i_\ell i_1\in \F-\fout$. Let $f_1$ and $f_2$ denote the unique faces in $\F-\fout$ that contain $\{i-1,i\}$ and $\{j,j+1\}$, respectively. Then,
\noindent
\begin{enumerate}[label=(\alph*),topsep=0.2ex]
\item  $\delta_G(\closedleft i,j \closedright)= G^*_{f_1,f_2}+\{i-1,i\}+\{j,j+1\}$.\label{lemma:dualpathwithEndpointsinOuter}
\item  $\delta_G(\closedleft i,j \closedright)\subseteq G^*_{f_1,f}\cup G^*_{f,f_2} +\{i-1,i\}+\{j,j+1\}$.\label{lemma:dualpathContainedIn}
\item $\delta_G(\closedleft i,j \closedright)=G^*_{f_1,f}\cup G^*_{f,f_2} +\{i-1,i\}+\{j,j+1\}$ if both $[i,j]$ and $[j+1,i-1]=[n]\setminus[i,j]$ contain a node of $f$.\label{lemma:fcontainedinPath} 
\end{enumerate}
\end{lemma}
\begin{proof}
By Theorem~\ref{theorem:dualofplanarGcycles}, $\delta_G([i,j])$ corresponds to the set of edges of a simple circuit $C$ in the dual graph of $G$. Since $G^*$ is a tree, $\fout$ must be present in $C$. 
Thus, $C-\fout$ is a path $G^*_{f^\prime_1,f^\prime_2}$ whose endpoints $f^\prime_1,f^\prime_2$ are faces adjacent to $\fout$. Observe that $\{i-1,i\}$ and $\{j,j+1\}$ are the only edges in $\delta_G([i,j])$ that are incident to $\fout$, and therefore, $\{f_1,f_2\}=\{f_1^\prime,f_2^\prime\}$. It follows that $G^*_{f_1,f_2}=\delta_G([i,j])-\{i-1,i\}-\{j,j+1\}$. This proves part~\ref{lemma:dualpathwithEndpointsinOuter}. Since $G^*$ is a tree, $G^*_{f_1,f_2}\subseteq G^*_{f_1,f}\cup G^*_{f,f_2}$. The statement of part~\ref{lemma:dualpathContainedIn} then follows from part~\ref{lemma:dualpathwithEndpointsinOuter}.

To prove part \ref{lemma:fcontainedinPath}, it suffices to show that $f\in G^*_{f_1,f_2}$, since that would imply, by the fact that $G^*$ is a tree, that $G^*_{f_1,f_2}=G^*_{f_1,f}\cup G^*_{f,f_2}$; in which case the result then follows from part~\ref{lemma:dualpathwithEndpointsinOuter}. If $f=f_1$ or $f=f_2$, then clearly $G^*_{f_1,f_2}=G^*_{f_1,f}\cup G^*_{f,f_2}$. Assume that $f_1\neq f\neq f_2$. Then neither $\{i-1,i\}$ nor $\{j,j+1\}$ are incident to $f$. Since both $[i,j]$ and $[j+1,i-1]$ contain a node of $f$, there is an edge $e$ of $f$ in $\delta_G(\closedleft i,j \closedright)$,
and since $e\neq\{i-1,i\}$ and $e\neq\{j,j+1\}$, we get from part~\ref{lemma:dualpathwithEndpointsinOuter} that $e\in G^*_{f_1,f_2}$. Therefore $f\in G^*_{f_1,f_2}$, as desired.
\end{proof}

%The proof of the above can be found in \cref{proofapp_i_j}.

As mentioned earlier, our plan is to first transform the given multicommodity flow instance into a more suitable instance consisting of only good demands, which allows us to invoke Theorem~\ref{theorem:G+Hplanar} in order to prove our main result.
\begin{restatable}{rethm}{theoremnew}\label{theorem:Newinstance}
Given an outerplanar graph $G=(V,E)$ and a multicommodity flow instance $\I(G,c,D,d)$ such that the cut-condition is satisfied, there exists a polynomial time algorithm that computes an instance $\I(G,c',D',d')$ and a partition $\{\W_g\}_{g\in D}$ of $D'$ such that the following properties hold.
\begin{enumerate}[label={\textnormal{(P\arabic*)}}, topsep=0.2ex, noitemsep, leftmargin=*]
\item \textnormal{(Feasibility)} The cut-condition is satisfied for $\I(G,c',D',d')$.
\label{cutcondition}
\item \textnormal{(Capacity)} $c'(e)\leq c(e)+\alpha_{\textsc{RL}}\cdot\Dmax$ for $e\in \fout$ and $c'(e)\leq c(e)+2 \cdot \alpha_{\textsc{RL}}\cdot\Dmax$ for $e\in E\setminus \fout$.\label{capproperty}
\item \textnormal{(Partition-into-paths)} For each $g \in D$, all the edges in $\W_g$ are \emph{good} and form a path between the endpoints of $g$. Furthermore, $d'(h)=d(g)$ for each $h \in \W_g$.\label{pathproperty}
\end{enumerate}
\end{restatable}
We devote the next subsections to describing the proof of Theorem~\ref{theorem:Newinstance}. 
Before that, we see how Theorems~\ref{theorem:G+Hplanar} and \ref{theorem:Newinstance} together imply Theorem~\ref{theorem:outerplanar}, which we restate.
\theoremmain*
\begin{proof} 
Given an instance $\I(G,c,D,d)$ satisfying the conditions of the theorem, we construct a feasible instance $\I(G,c',D',d')$ using Theorem~\ref{theorem:Newinstance}.
The (Partition-into-paths) property \ref{pathproperty} in Theorem~\ref{theorem:Newinstance} implies that every demand in $D'$ is good. 
It then follows from Theorem~\ref{theorem:G+Hplanar} that $\I(G,c',D',d')$ admits a $1$-feasible unsplittable flow $z$. For a demand $h \in D'$, let $z(h)$ be the edge set of the path on which $z$ routes $h$. We construct an unsplittable flow $y$ for instance $\I(G,c,D,d)$ as follows: for $g \in D$, let $p_g$ be a path between the end-points of $g$ using the edges in $\cup_{h \in W_g} z(h)$. The (Partition-into-paths) property \ref{pathproperty} ensures that such a path $p_g$ always exists. We set $y_g(p_g)=d(g)$ and $y_g(p)=0$ for all $p \in \mathcal{P}_{g} \setminus \{p_g\}$. For each $g \in D$ and $h \in W_g$, The (Partition-into-paths) property \ref{pathproperty} states that $d'(h)=d(g)$. Hence, the following holds for all $e\in E$: 
%and hence it follows that $y(e) \leq z(e)$ for all $e \in E$. For all $e \in E$, we have:
\begin{equation*}
 y(e) \leq z(e) \leq c'(e)+d_{\max} \leq c(e) + 2 \cdot \alpha_{\textsc{RL}}\cdot\Dmax+ d_{\max}=c(e) + (2 \cdot \alpha_{\textsc{RL}}+1) \cdot d_{\max}.
\end{equation*}
The third inequality follows from the (Capacity) property ~\ref{capproperty}.
%in Theorem~\ref{theorem:Newinstance}.
We conclude that $y$ is a ($2 \cdot \alpha_{\textsc{RL}}+1)$-feasible unsplittable flow for $\I(G,c,D,d)$, proving the theorem. 
\end{proof}
We now describe our main algorithm which will lead to a proof of Theorem~\ref{theorem:Newinstance}.  
%\vspace*{-1ex}
\subsection{Description of the Algorithm}\label{descriptionalgo}
%\vspace*{-1ex}
\subsubsection{Overview:}
In this section, we give a detailed description of the algorithm. The input of the algorithm is an instance $\I(G,c,D,d)$ of multicommodity flow satisfying the cut-condition. We also assume that we are given an algorithm $\Alg$ which given a feasible ring-loading instance outputs an $\alpha$-feasible unsplittable flow in polynomial time.
%\david{(is the last sentence unambiguous?)}.

The algorithm works in iterations. At the end of every iteration, we use a \emph{pinning} operation to create a new (feasible) instance. Recall that a demand $h$ is \emph{bad} if no edge of $G$ has the same end-points as $h$. In each iteration, we pin a carefully chosen set of bad demands and increase the capacity of carefully chosen edges to create a new feasible instance. The algorithm terminates when there are no bad demands left. 

To keep track of the bad demands and edge capacity increases, we use edge-labels. The label of an edge changes during the course of the algorithm and takes value in $\{0,1,2\}$. Initially, all edges on the unbounded face $\fout$ are assigned a label of 1, while the rest of the edges are assigned a label of 0. The label of an edge cannot decrease during the course of the algorithm. The algorithm terminates when all the edges have a label of 2. The label of an edge keeps track of the increase in capacity of an edge as follows: whenever the label of an edge $e$ increases by one, we increase the capacity of $e$ by $\alpha \cdot \Dmax$. Since the label of any edge can only increase by at most 2, we do not increase the capacity of any edge by more than $2 \cdot \alpha \cdot \Dmax$ during the course of the algorithm. 

We now list some useful properties of edge labels, which are also summarized in Lemma~\ref{propositionCycle2}. 
%The edges with label 1 form a simple cycle and the label of an edge is 0 if and only if both of its end-points are on the cycle formed by label 1 edges.
The edges with label 1 form a simple cycle $C$ and the label of an edge $e\in E$ is 0 if and only if both of its end-points are on $C$ and $e\notin C$. All the edges which do not have label 1 or 0 have a label of 2. The algorithm terminates when no edge has label 1. This implies that there are no edges of label 0 and all the edges have label 2 at the termination. It is easy to see that these invariants hold at the beginning of the algorithm.

The cycle bounded by label 1 edges is used to keep track of the bad demands. More precisely, both the end points of a bad demand are incident on the cycle formed by label 1 edges. At the termination of the algorithm, there are no label 1 edges and hence the final instance comprises only of \emph{good} demands.
%We next give a detailed description of an iteration of the algorithm.
%\vspace*{-1ex}
\subsubsection{Description of an Iteration:}
At the start of the algorithm we define the following:
%\vspace*{-2ex}
\begin{equation*}
c^{(1)}:=c~;~D^{(1)}:=D~;~\W^{(1)}_g:=\{g\}~\text{for}~g\in D~;
~\labone(e):=\begin{cases}
1,&\:\text{if }~e\in\fout\\
0,&\:\text{otherwise.}
\end{cases}   
\end{equation*}
Let $x^{(1)}=\sum_{g\in D}x^{(1)}_g$ be a feasible flow of $\I(G,c,D,d)$.\footnote{By Theorem \ref{OS}, the cut-condition is sufficient for the existence of a feasible flow on outerplanar graphs.} We now give a detailed description of an iteration of the algorithm.
Suppose we are at the start of iteration $\ite$ of the algorithm. 
Let $\Dp$ be the current set of demands and $\{\Wk_g\}_{g \in D}$ be the partition of $\Dp$.
We note that the demand edges in $\W_g^{(k)}$ form a walk between the end points of $g$ for $g \in D$. 
Let $x^{(k)}$ be the feasible flow for the current instance $\I(G,\cp,\Dp,\dpr)$, 
and $\labk$ be the current labels of the edges. The current demand values $\dpr$ are always given by $\dpr(h)=d(g)$, if $h\in\Wk_g$.

If there are no label 1 edges, then the algorithm returns the instance $\I(G,c'=\cp,D'=\Dp,d'=\dpr)$ and the set $\{\W_g=\W^{(k)}_g\}_{g \in D}$, and terminates.

Otherwise, let $C^{(k)}$ be the (simple) cycle formed by the edges $e\in E$ with $\labk(e)=1$ and $G[C^{(k)}]$ be the induced graph on the vertices of $C^{(k)}$. Observe that $G[C^{(1)}]=G$.
We overload notation and use $C^{(k)}$ to denote both the edge and vertex set of the cycle $C^{(k)}$. 

Since $G[C^{(k)}]$ is 2-vertex connected and outerplanar, its weak dual $G^{*}[C^{(k)}]$ is a tree.
Also, the edges of $G^{*}[C^{(k)}]$ correspond precisely to the edges $e\in E$ with $\labk(e)=0$.
Let $f$ be a leaf vertex of $G^{*}[C^{(k)}]$. Note that $f$ corresponds to a bounded face of $G$.

If $G[C^{(k)}]=C^{(k)}$, then $f=C^{(k)}$ and $\labk(e)=1$ for each $e\in f$. 
Otherwise, $\labk(e^\prime)=0$ for precisely one edge $e^\prime\in f$ and $\labk(e)=1$ for each $e\in f-e^\prime$.
Thus, suppose that $f$ is given by the vertex sequence $f=i_1 i_2 \ldots i_\ell i_1$ such that $\ell \geq 3$ and the label of all edges of $f$ distinct from $\{i_\ell,i_1\}$ is 1. 
The label of edge $\{i_\ell,i_1\}$ could be 0 or 1.

At the end of the current iteration, we increase the label of all edges of $f$ by 1. Hence the only edge of $f$ which could have a label of 1 after this iteration is $\{i_\ell,i_1\}$. This implies that in iteration $(k+1)$, either there are no edges of label 1, or the cycle formed by label 1 edges is $(C^{(k)} \setminus E(f)) \cup\{ \{i_\ell,i_1\}\}$. As mentioned in the last section, we want to maintain the invariant that both the end-points of a bad demand are incident on the cycle formed by the label 1 edges. To this end, we will pin a set of demands so that no bad demands are incident on the vertices $\{i_2,i_3,\ldots,i_{\ell-1}\}$ after the end of the current iteration.

Let $\Dp_1\subseteq \Dp$ be the set of bad demands containing one end-point in $\{i_2,i_3,\ldots,i_{\ell-1}\}$ and one end-point in $C^{(k)} \setminus V(f)$.\footnote{By outerplanarity, any demand with an endpoint in $\{i_2,i_3,\ldots,i_{\ell-1}\}$ and an endpoint in $C^{(k)} \setminus V(f)$ is bad.} 
Let $\Dp_2\subseteq \Dp$ be the set of bad demands both whose end-points are incident on $\{i_1,i_2,\ldots,i_\ell\}$.
Define:
%\vspace*{-1ex}
\begin{equation}\label{eq:defx}
x^*:=\sum_{h\in\Dp_1\cup\Dp_2}x^{(k)}_h.
%\vspace*{-1ex}
\end{equation}
We would like to find an appropriate pinning of the demands in $\Dp_1 \cup \Dp_2$, and to this end, we create a ring-loading instance as follows. For $i\in[n]$, let $f_i$ be the unique face in $\F-\fout$ containing the edge $\{i,i+1\}$. We define the following ring-loading instance,
%\vspace*{-1ex}
\begin{equation}
\mathcal{RL}([n],\crl,\Dp_1\cup\Dp_2,\dpr)~\text{where}~\crl(\{i,i+1\}):=x^*\big(G^*_{f,f_i}\big)+x^*(\{i,i+1\}).\label{ringLoadfeasible}
%\label{definingEdgecapacitiesforRL}  
\end{equation}
See Figure \ref{ringloadingdrawing} for an example. In Lemma~\ref{lemma:instancesSatisfyCutCond} we show that the cut-condition is satisfied for this ring-loading instance, and hence we can use $\Alg$
to obtain an $\alpha$-feasible unsplitable flow $z=\sum_{h\in \Dp_1\cup\Dp_2}z_h$. 
\begin{figure}[t]%{0.5\textwidth}
        \includegraphics[width=1\linewidth]{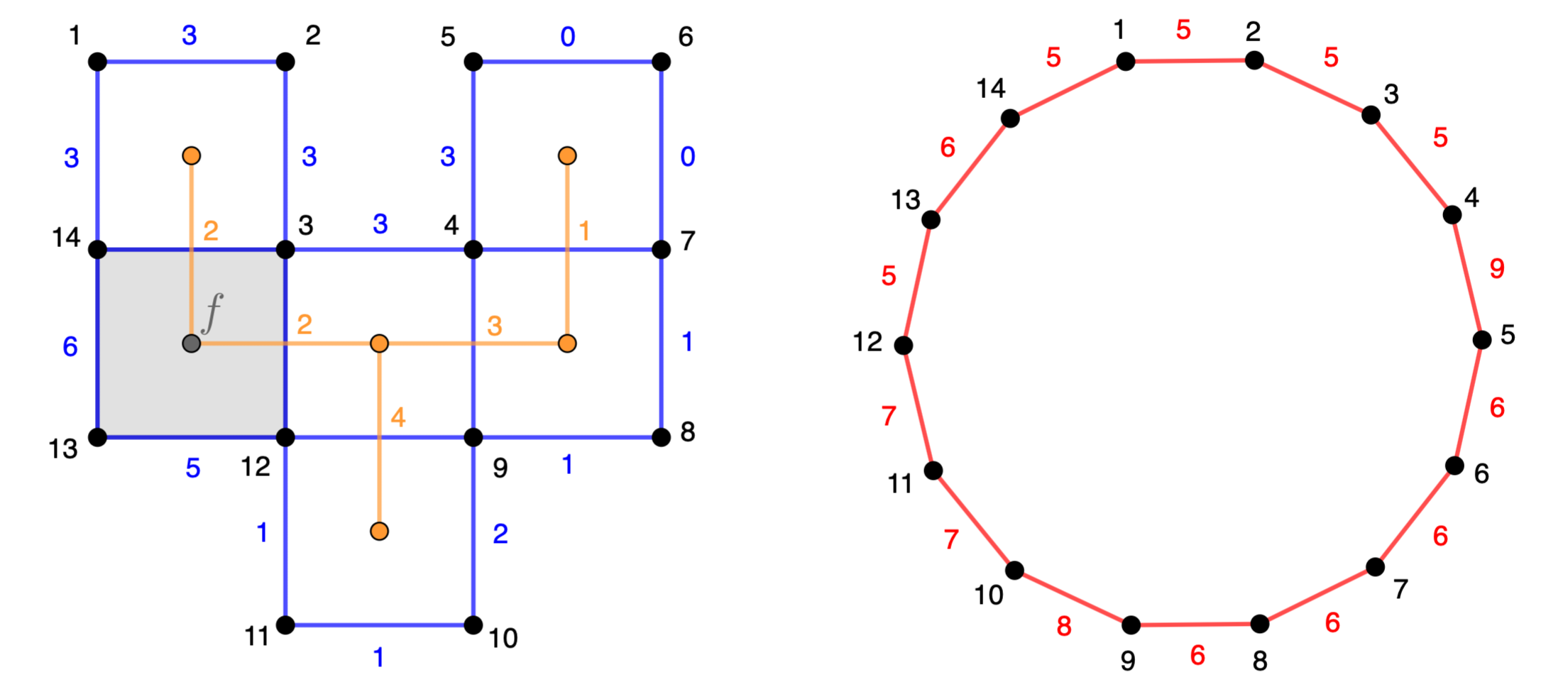}
        \vspace*{1ex}
       \caption{ The figure on the left depicts an outerplanar graph $G$ where each edge $e\in E$ is labelled by $x^*(e)$, where $x^*$ is defined as in \eqref{eq:defx}. 
        The shaded region represents the face $f$ being considered during the iteration, and the yellow tree corresponds to the weak dual $G^*$ of $G$.
        The figure on the right depicts the resulting capacities $\crl$ assigned to the ring loading instance $\mathcal{RL}([n],\crl,\Dp_1\cup\Dp_2,\dpr)$. } \label{ringloadingdrawing}
        %\vspace*{-2ex}
\end{figure}
%and create a ring-loading instance $\mathcal{RL}([n],\crl,\Dp_1\cup\Dp_2,\dpr)$\label{ringLoadfeasible}.
For $h \in \Dp_1 \cup \Dp_2$, we overload notation and use $z_h$ to also denote the (unique) path along which $h$ is routed in $z$. 
The manner in which we pin a demand $h$ depends on whether the path $z_h$ uses the vertices $i_1,i_\ell$ or not. 
%\david{(good spot for a figure if time allows)}
More precisely, we replace each demand $h \in \Dp_1 \cup \Dp_2$ by a set of demands $\Dp_h$ as described below.

\noindent\textbf{Case 1: \label{FirstForstarts}} Suppose that $h\in\Dp_1$ and $V(h)=\{s,t\}$; without loss of generality, we may assume that $s=i_p$ for some $1<p< \ell$, and $t\in C^{(k)} \setminus V(f)$. If $i_1 \in z_h$, then
%\vspace*{-1ex}
\begin{equation}
  \Dp_h := \{\{i_j,i_{j+1}\}:1\leq j< p\}\cup\{\{i_1,t\}\}.\label{1stD1option}  
\end{equation}
Otherwise, if $i_\ell \in z_h$, then
%\vspace*{-1ex}
\begin{equation}
  \Dp_h :=\{\{i_j,i_{j+1}\}:p\leq j< \ell\}\cup\{\{i_\ell,t\}\}.\label{2ndD1option}  
\end{equation}
\noindent\textbf{Case 2:\label{SecondForStarts}} Now suppose that $h \in \Dp_2$ and $V(h)=\{s,t\}$; without loss of generality, we may assume that $s=i_p$ and $t=i_q$ for $1\leq p< q\leq\ell$. If both $i_1,i_\ell \in z_h$, then: 
%\vspace*{-1ex}
\begin{equation}
\Dp_h := \{\{i_j,i_{j+1}\}:1\leq j< p\}\cup\{\{i_\ell,i_1\}\}\cup\{\{i_j,i_{j+1}\}:q\leq j< \ell\}. \label{2ndD2option}
\end{equation}
%\david{(Should we either write $\Dp_h := \{\{i_j,i_{j+1}\}:1\leq j< p\}\cup\{\{i_j,i_{j+1}\}:q\leq j< \ell\}+\{i_\ell,1\}$, or add another bracket to $\{i_\ell,i_1\}$ in the above equation?)}
Otherwise; if either $i_1\notin z_h$ or $i_\ell \not \in z_h$,\footnote{It may be the case that $|\{i_1,i_\ell\}\cap V(z_h)|=1$, if either $p=1$ or $q=\ell$.}
then:
%\vspace*{-1ex}
\begin{equation}
\Dp_h := \{\{i_j,i_{j+1}\}:p\leq j< q\}. \label{1stD2option}   
\end{equation}
In \cref{pr3} of Lemma~\ref{propositionCycle2} we show that there is at most one bad demand edge in $\W_g^{(k)}$. If $h \in \W_g^{(k)}$ for $h \in \Dp_1 \cup \Dp_2$, then we update $\Wk_g$ by removing $h$ and adding $\Dp_h$. 
%\vspace*{-1ex}
\begin{alignat}{2}
\Wkplus_g&\leftarrow (\Wk_g-h)\cup\Dp_h&&,  \hspace{3 mm}\text{if there exists $h\in \Wk_g$ with $h\in\Dp_1\cup\Dp_2$ }; \nonumber\\
\Wkplus_g&\leftarrow \Wk_g&&, \hspace{3 mm}\text{otherwise.}\label{updateWg}
\end{alignat}
We update $\Dp$ as follows: 
\begin{equation}
    \Dpplus\leftarrow \cup_{g\in D}\Wkplus_g \label{updateDp}
\end{equation}
We also update $d^{(k+1)}$ accordingly. 
Let $\bon_f\in\{0,1\}^E$ be the characteristic vector of the edge set of $f$. 
In Lemma~\ref{lemma:instancesSatisfyCutCond}, we show that the cut-condition is satisfied for the following instance:
\begin{equation}
  \mathcal{I}\left(G,x^*+\alpha \cdot \Dmax\cdot\bon_f,D_f,d^{(k+1)}\right) \hspace{3 mm} \text{where} \hspace{3 mm}D_f:=\cup_{h\in \Dp_1\cup\Dp_2}\Dp_{h}. \label{cutconditionEndofIteration} 
\end{equation}
We compute a feasible flow $y=\sum_{r\in D_f}y_r$ for this instance and set,
\begin{equation}
 x^{(k+1)}\leftarrow x^{(k)}-x^*+y \label{newflowx}   
\end{equation}
%\david{(undefined $x^\prime$)}
As already mentioned, we increase the label of all the edges of $f$ by 1, and  also increase their capacity by $\alpha \cdot \Dmax$. More precisely,
\begin{equation}
c^{(k+1)} \leftarrow \cp+\alpha \cdot \Dmax\cdot\bon_f \hspace{3 mm} \text{and} \hspace{3 mm} \lab^{(k+1)}\leftarrow \labk+\bon_f. \label{updateflowcapacityandlabel}   
\end{equation}
This completes the description of the algorithm. 
See Figure \ref{images_algorithm} for an example that illustrates how the demand edges could evolve on each iteration of the algorithm.
%\newpage
\begin{figure}[ht!]
    \centering
    \vspace*{1.5ex}
    \begin{subfigure}{1\textwidth}
    \includegraphics[width=0.95\linewidth]{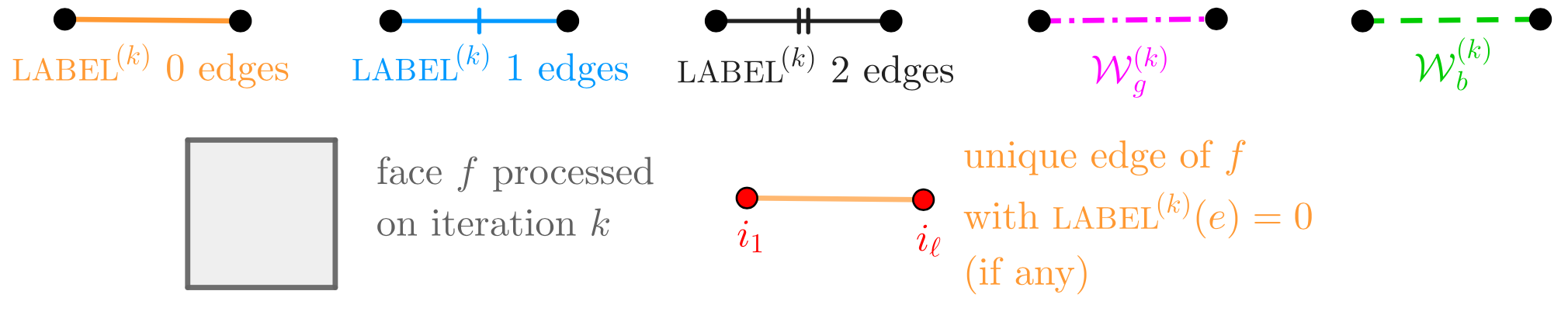}
    \centering
    \end{subfigure}
    \vspace*{2ex}
    \begin{subfigure}{0.5\textwidth}
        \includegraphics[width=0.9\linewidth]{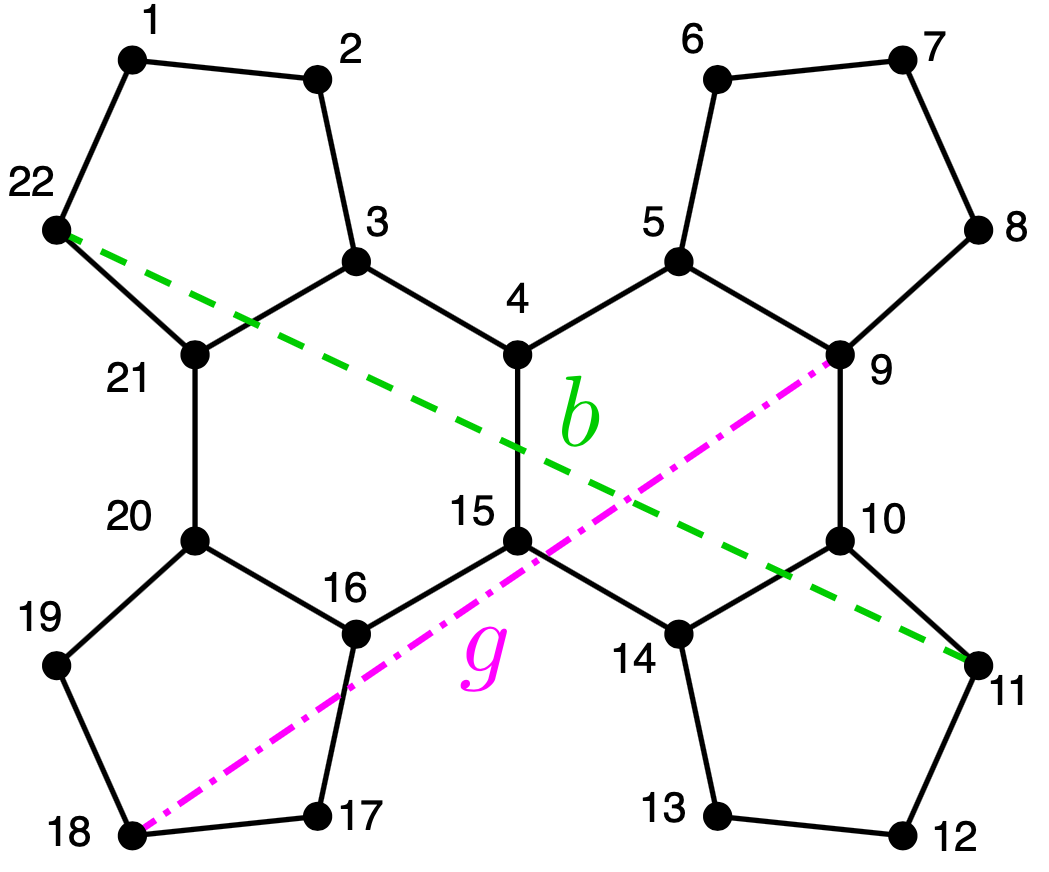}
        \centering
        %\vspace*{-0.22cm}
        \caption*{{\small Original instance with $D=\{\textcolor{verde}{b},\textcolor{morado}{g}\}$}}
    \end{subfigure}
    \vspace*{2ex}
    \hspace*{-3ex}
    \begin{subfigure}{0.5\textwidth}
        \includegraphics[width=0.9\linewidth]{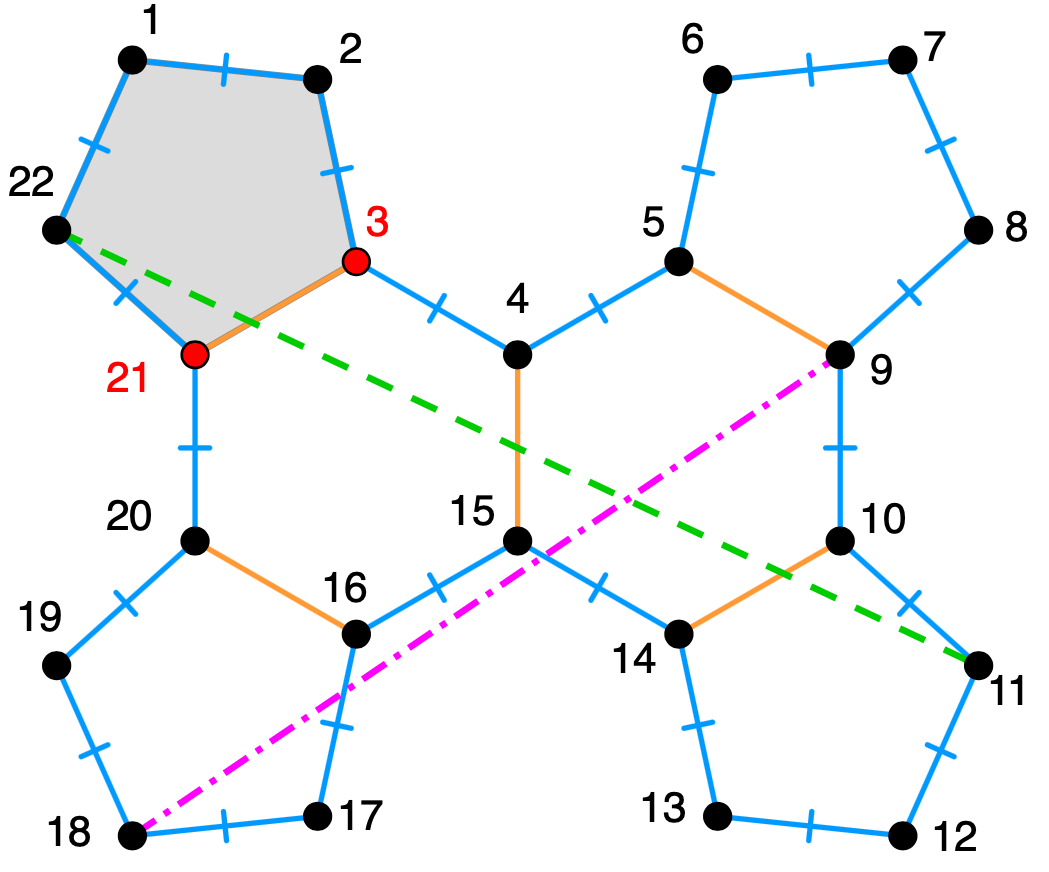}
        \centering
        \caption*{{\small \textbf{Iteration 1:} $D^{(1)}_1=\{\textcolor{verde}{\{11,22\}}\}$, $D^{(1)}_2=\emptyset$.}}
    \end{subfigure}
    \vspace*{2ex}
    \begin{subfigure}{0.5\textwidth}
        \includegraphics[width=0.9\linewidth]{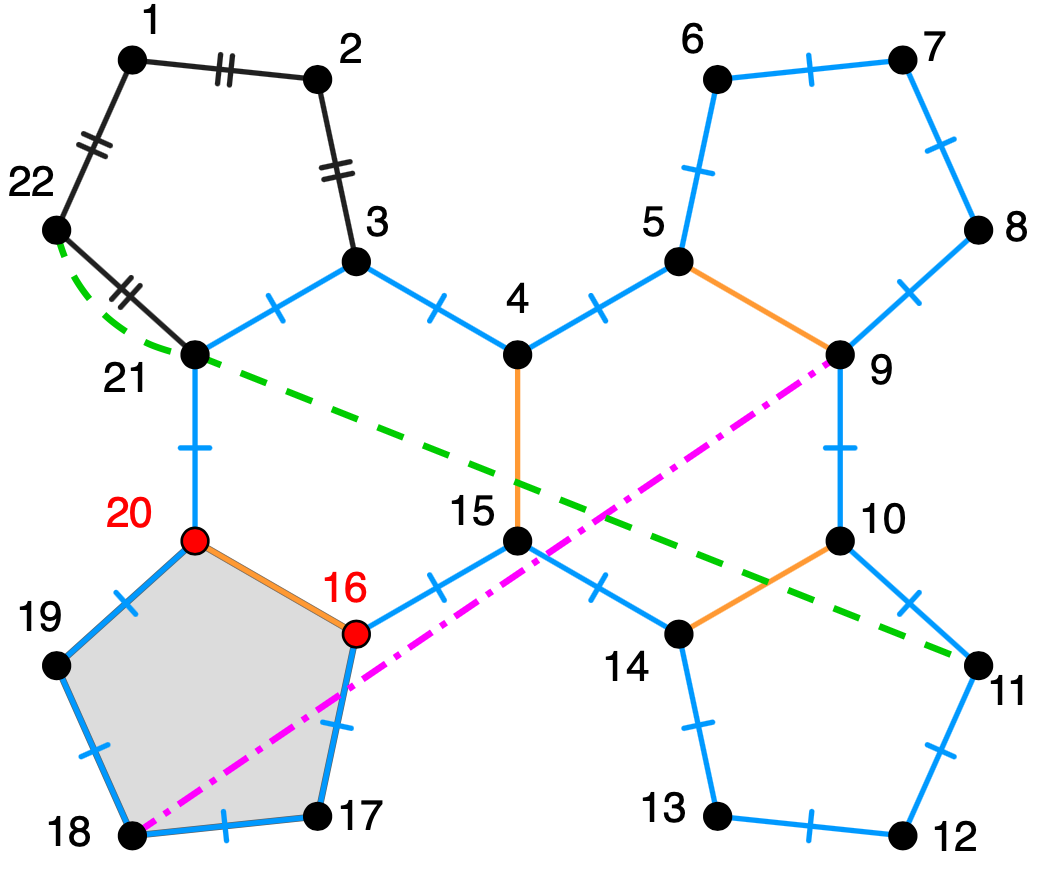}
        \centering
        \caption*{ {\small \textbf{Iteration 2:} $D^{(2)}_1=\{\textcolor{morado}{\{9,18\}}\}$, $D^{(2)}_2=\emptyset$.}}
    \end{subfigure}
    \vspace*{2ex}
    \hspace*{-3ex}
    \begin{subfigure}{0.5\textwidth}
        \includegraphics[width=0.9\linewidth]{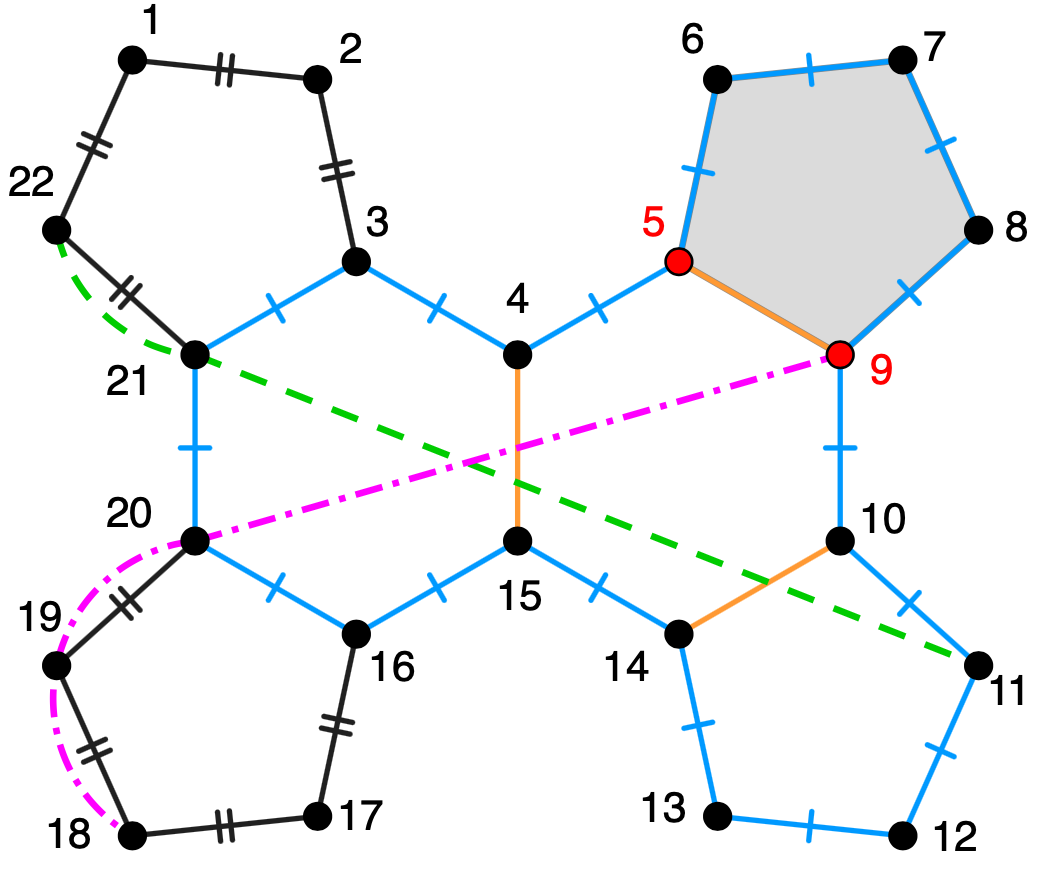}
        \centering
        \caption*{ {\small \textbf{Iteration 3:} $D^{(3)}_1=\emptyset$, $D^{(3)}_2=\emptyset$.}}
    \end{subfigure}
    % \vspace*{1ex}
    % \begin{subfigure}{0.5\textwidth}
    %     \includegraphics[width=0.8\linewidth]{ite4f.png}
    %     \centering
    %     \caption*{ {\small \textbf{Iteration 4:} $D^{(4)}_1=\{\textcolor{verde}{\{11,21\}}\}$, $D^{(4)}_2=\emptyset$.}}
    % \end{subfigure}
    % %\vspace*{2ex}
    % \hspace*{-3ex}
    % \begin{subfigure}{0.5\textwidth}
    %     \includegraphics[width=0.8\linewidth]{ite5f.png}
    %     \centering
    %     \caption*{ {\small \textbf{Iteration 5:} $D^{(5)}_1=\{\textcolor{verde}{\{14,21\}},\textcolor{morado}{\{9,20\}}\}$, $D^{(5)}_2=\emptyset$.}}
    % \end{subfigure}
     %\vspace*{2ex} 
    \end{figure}
 %\newpage   
    \setcounter{figure}{1}
    \begin{figure}[ht!]
     \begin{subfigure}{0.5\textwidth}
        \includegraphics[width=0.9\linewidth]{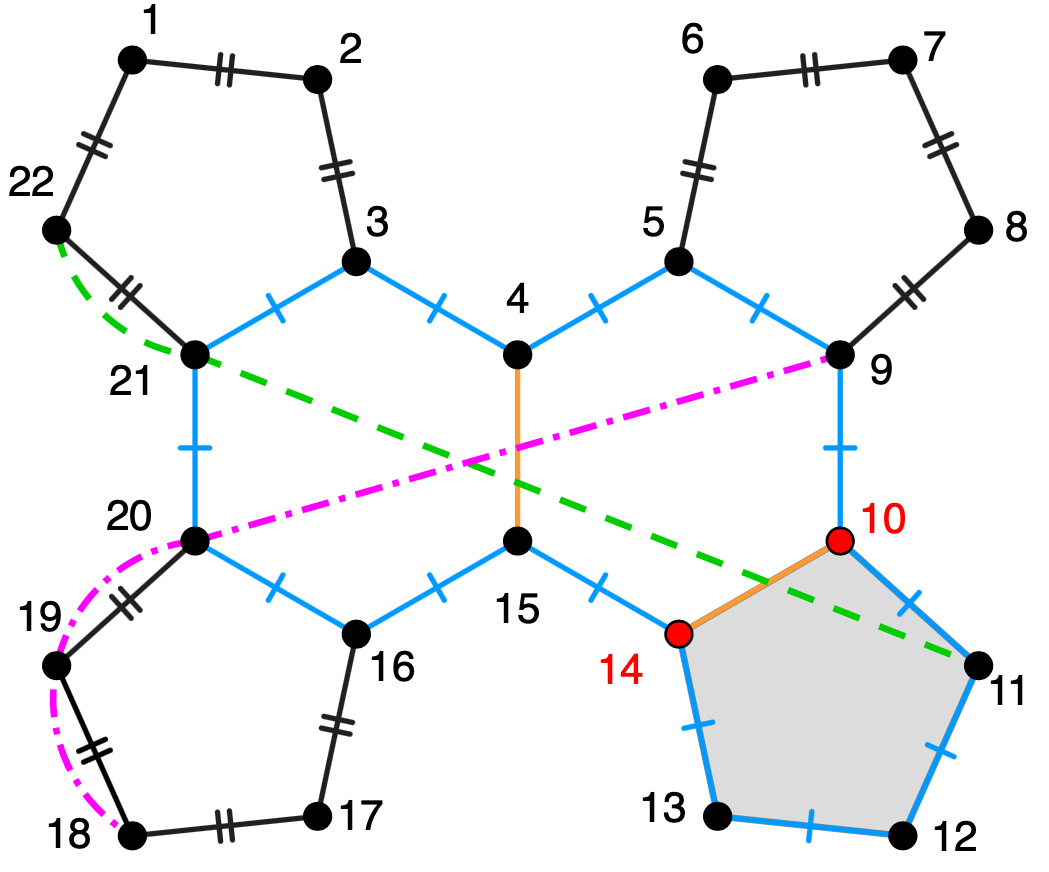}
        \centering
        \caption*{ {\small \textbf{Iteration 4:} $D^{(4)}_1=\{\textcolor{verde}{\{11,21\}}\}$, $D^{(4)}_2=\emptyset$.}}
    \end{subfigure}
    \vspace*{2ex}
    \hspace*{-3ex}
    \begin{subfigure}{0.5\textwidth}
        \includegraphics[width=0.9\linewidth]{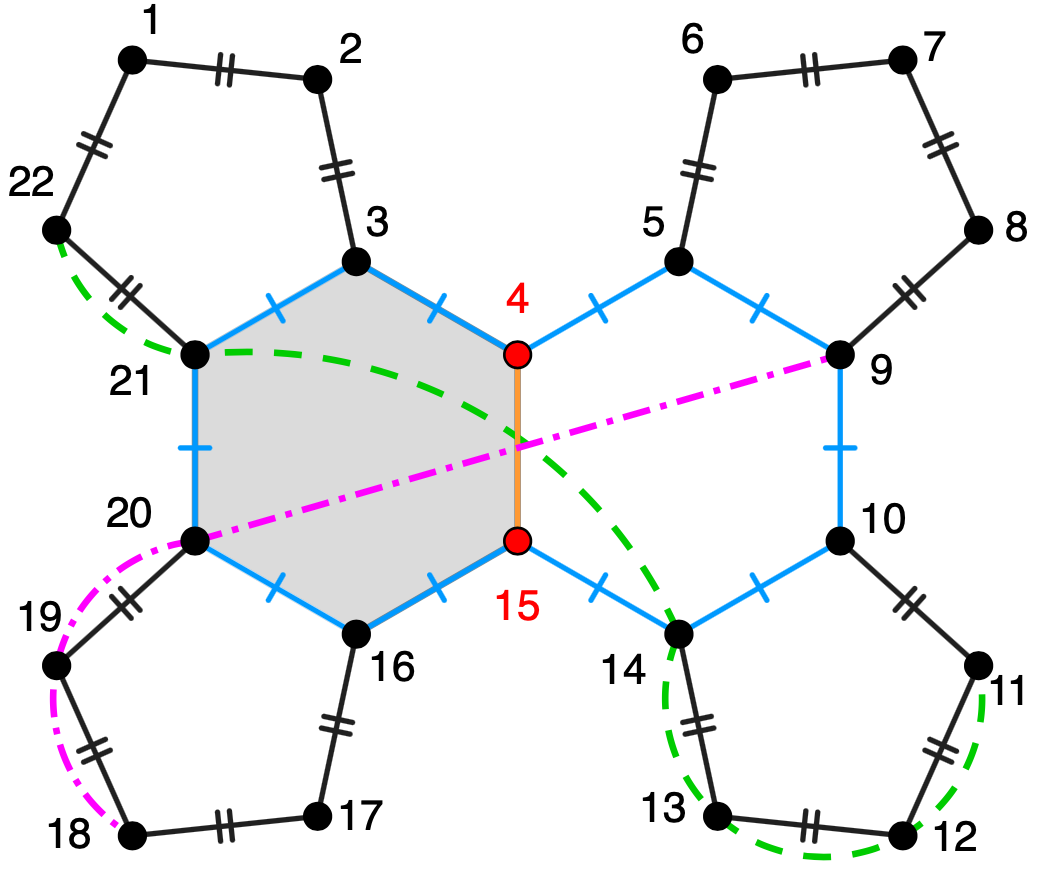}
        \centering
        \caption*{ {\small \textbf{Iteration 5:} $D^{(5)}_1=\{\textcolor{verde}{\{14,21\}},\textcolor{morado}{\{9,20\}}\}$, $D^{(5)}_2=\emptyset$.}}
    \end{subfigure}

    \begin{subfigure}{0.5\textwidth}
        \includegraphics[width=0.9\linewidth]{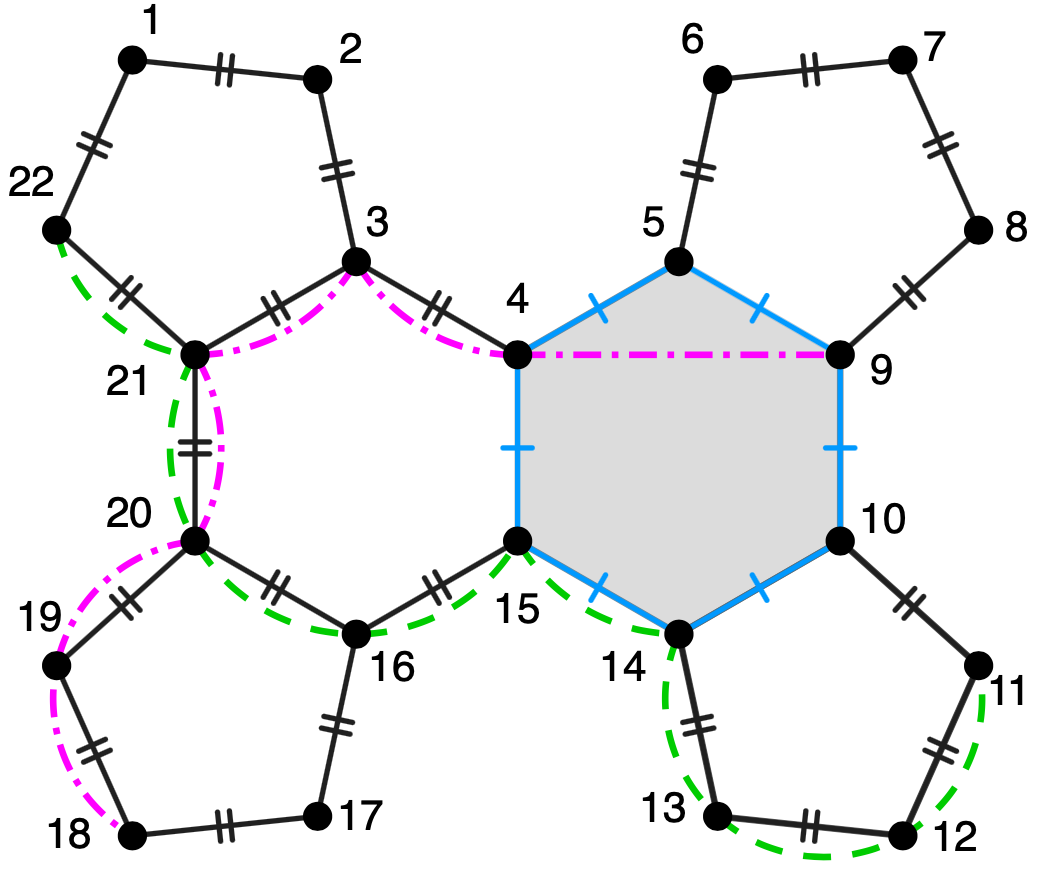}
        \centering
        \caption*{ {\small \textbf{Iteration 6:} $D^{(6)}_1=\emptyset$, $D^{(6)}_2=\{\textcolor{morado}{\{4,9\}}\}$.}}
    \end{subfigure}
    \vspace*{2ex}
    \hspace*{-3ex}
    \begin{subfigure}{0.5\textwidth}
        \includegraphics[width=0.9\linewidth]{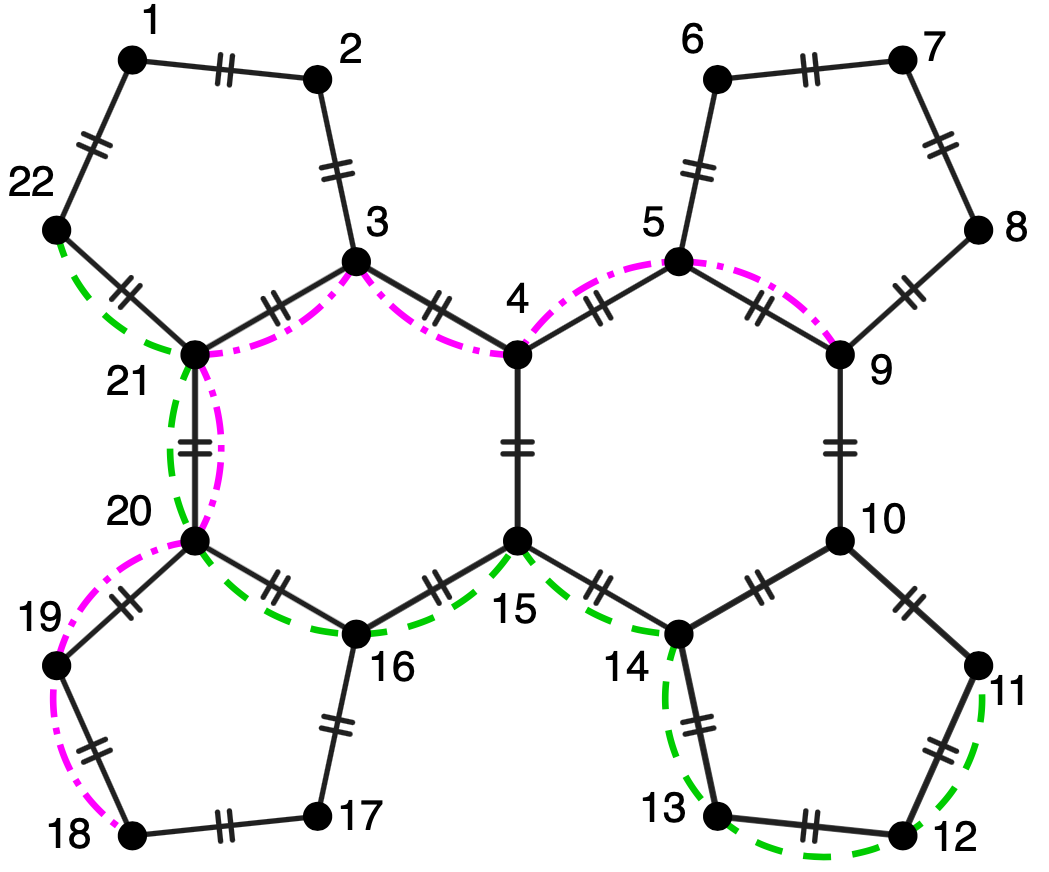}
        \centering
        \caption*{{\small \textbf{Iteration 7:} $\textsc{label}^{\hspace{-0.07cm}(7)}(e)=2$ for each $e\in E$. }}
    \end{subfigure}
    \vspace{1ex}
     \caption{The algorithm terminates at the start of iteration 7. 
     The demand edges of the resulting instance returned by the algorithm are $D^{(7)}=\W^{(7)}_b\cup\W^{(7)}_g$, where\\
    $\W^{(7)}_b=\{\textcolor{verde}{\{11,12\},\{12,13\},\{13,14\},\{14,15\},\{15,16\},\{16,20\},\{20,21\},\{21,22\}}\}$,\\
    $\W^{(7)}_g=\{\textcolor{morado}{\{9,5\},\{5,4\},\{4,3\},\{3,21\},\{21,20\},\{20,19\},\{19,18\}}\}$.}\label{images_algorithm}
    %\label{fig:enter-label22}
    %\vspace*{-3ex}
\end{figure}

\subsection{Proof of Correctness}
\noindent We first note a few key lemmas which are useful in proving Theorem~\ref{theorem:Newinstance}. 
For any iteration $\ite\geq 1$ and $j\in\{0,1,2\}$, 
let $E_j(k):=\{e\in E:\labk(e)=j\}$.
Recall that $\labk(e)$ denotes the label of $e$ at the start of iteration $\ite$;
where we initially define $\labone(e)=1$ if $e\in\fout$, and $\labone(e)=0$ if $e\in E\setminus E(\fout)$.

\begin{lemma}\label{propositionCycle2}
The following holds for each iteration $\ite\geq 1$ and $g\in D$.
\begin{enumerate}[label=(\alph*),topsep=0.2ex]
\item $E_1(k)$ is either the edge-set of a simple cycle $C^{(k)}$ in $G$, or the empty set.\label{pr1}
\item Consider the induced graph $G[C^{(k)}]$ on the vertex-set of $C^{(k)}$. Then $E_0(k)$ are precisely the edges of $G[C^{(k)}]$ that do not lie on the cycle $C^{(k)}$, and $E_2(k)$ are precisely the edges of $G$ that are not contained in $G[C^{(k)}]$.\label{pr2}
\item The demand edges in $\W_g^{(k)}$ form a walk between the end-points of $g$. Furthermore, there is at most one bad demand edge $h\in \W_g^{(k)}$.
If such $h$ exists, then both of its end-points lie on the cycle $C^{(k)}$.\label{pr3}
\end{enumerate}
\end{lemma}
\begin{proof}
The proof is by induction on $k$. 
The result clearly holds when $k=1$.
Suppose that $k\geq 2$. By the induction hypothesis, the result holds for $k-1$. We consider three cases.\\

\noindent \textbf{Case 1:} $E_1(k-1)=\emptyset$. 

By \ref{pr2}, it follows that $\labkminus(e)=2$ for each $e\in E$; thus, the algorithm reaches its termination-condition at the start of iteration $(k-1)$ (in particular there is no iteration $\ite$).\\
%thus, no face of $G$ satisfies the conditions of the while-loop after iteration $k-1$, and the algorithm terminates after iteration $(k-1)$.\\  

\noindent \textbf{Case 2:} $E_1(k-1)\neq\emptyset$ and $G[C^{(k-1)}]=C^{(k-1)}$. 

Then the cycle $C^{(k-1)}$ is a face of $G$. By \ref{pr2}, it follows that $\labkminus(e)=2$ for all edges $e\in E\setminus E_1(k-1)$.
%Thus $C^{(k-1)}$ is the face processed in iteration $k$. 
Thus $C^{(k-1)}$ is the (unbounded) face $f$ of $G$ processed in iteration $k-1$.
%Since $\labkminus(e)=1$ for each $e\in C^{(k-1)}$, then $\labk(e)=2$ for each $e\in C^{(k-1)}$, and hence $\labk(e)=2$ for each $e\in E$. 
Since $\labk =\labkminus +\bon_f$ (see \eqref{updateflowcapacityandlabel}), it follows that $\labk(e)=2$ for each $e\in E$.
This proves \ref{pr1} and \ref{pr2}. 
We now show \ref{pr3}. 
By the induction hypothesis, (i) there is at most one bad demand edge $h\in\Wkminus_g$; and
(ii) $\Wkminus_g$ is the edge-set of a walk between the endpoints of $g$.
%If such an edge does not exist, then $\W_g^{(k-1)}=\W_g^{(k)}$ since $\Dp_1$ and $\Dp_2$ (defined on lines \ref{definingDP1} and \ref{definingDP2} of \ref{algorithmFaceAug}) only contain edges whose endpoints are not adjacent in $G$, and \ref{pr3} is trivially satisfied in this case. 
If such a bad demand edge $h$ does not exist, then $\W_g^{(k-1)}=\W_g^{(k)}$ (see \eqref{updateWg}), and therefore \ref{pr3} is trivially satisfied in this case.
Suppose on the other hand that there is a bad demand edge $h\in\Wkminus_g$.
By \ref{pr3}, 
both end-points $s,t$ of $h$ lie on $C^{(k-1)}$; 
therefore $h \in \Dpminus_2$.
%\david{(should we place the definition of $\Dp_1$ and $\Dp_2$ on an equation on the previous section?)}
At the end of iteration $k-1$, 
edge $h$ is replaced by a sequence of demand edges $\Dpminus_h$ which form a simple $s-t$ path along the edges of $C^{(k-1)}$ (see \eqref{2ndD2option} and \eqref{1stD2option}).
Thus, (i) $\Wk_g=(\Wkminus_g-h)\cup \Dpminus_h$ only contains \emph{good} demand edges; and (ii) $\Wk_g$ is the edge-set of a walk between the end-points of $g$.
This proves \ref{pr3}.\\

\noindent \textbf{Case 3:} $E_1(k-1)\neq\emptyset$ and $G[C^{(k-1)}]\neq C^{(k-1)}$. 

Then $C^{(k-1)}$ is not a bounded face\footnote{Recall that all the faces of $G$ except $\fout$ are bounded} of $G$. Also, $G[C^{(k-1)}]$ is a (2-vertex connected) outerplanar graph and each bounded face of $G[C^{(k-1)}]$ is also a bounded face of $G$. 
By Fact \ref{fact:weakdualtree}, the weak dual $G^*[C^{(k-1)}]$ of $G[C^{(k-1)}]$ is a tree. 
By \ref{pr2}, the edge-set $E_0(k-1)$ is precisely the edge-set of $G^*[C^{(k-1)}]$;  
%therefore, the faces of $G$ satisfying the conditions of the while-loop on iteration $k$ are precisely the faces representing the leaf nodes of $G^*[C^{(k-1)}]$.
therefore, the faces of $G$ that can be processed on iteration $k-1$ are precisely the faces representing the leaf nodes of $G^*[C^{(k-1)}]$. 
%Recall that the face $f$ that is processed during an iteration of the algorithm is one for which the 
Let $f=i_1i_2\hdots i_\ell i_1$ be the face that is processed on iteration $\ite-1$. 
Then $\labkminus(\{i_j,i_{j+1}\})=1$ for each $1\leq j\leq \ell-1$ and $\labkminus(\{i_\ell,i_{1}\})=0$. 
Let $\omega_1\omega_2\hdots \omega_p i_1i_2\hdots i_\ell \omega_1$ denote the node-sequence of the simple cycle $C^{(k-1)}$, 
where $p\geq 1$. 
%Since $\lab(k)=\lab(k-1)+\bon_f$,
Since $\labk=\labkminus+\bon_f$,
the edge-set $E_1(k)$ induces the simple cycle $C^{(k)}$ given by the node-sequence $\omega_1\omega_2\hdots \omega_p i_1i_\ell \omega_1$. 
This proves \ref{pr1}. 
By the above, the only edge $e\in E$ such that $\labkminus(e)=0$ and $\labk(e)=1$ is $e=\{i_\ell,i_1\}$ i.e.,~$E_0(k)=E_0(k-1)-\{i_\ell,i_1\}$.
Therefore, $E_0(k)$ consists precisely of the edges $e$ of $G$ whose two endpoints lie on $C^{(k)}$ and $e\notin C^{(k)}$. 
This proves \ref{pr2}. 
    % Observe that $E_1(k-1)$ induce the simple cycle $C(k-1)$ given by a node-sequence $\omega_1\omega_2\hdots \omega_p i_1i_2\hdots i_\ell \omega_1$, where $p\geq 1$ since $C(k-1)$ is not a face of $G$; 
    % therefore, $E_1(k)$ induce the simple cycle $C(k)$ given by the node-sequence $\omega_1\omega_2\hdots \omega_p i_1i_\ell \omega_1$. This proves \ref{pr1}.     
    % By $(i)$ and $(ii)$, the edge-set $E_0(k)$ are precisely the edges of $G^*[C(k-1)]-f$,
    % which is precisely the weak dual graph of $G[C(k)]$; and therefore, $E_0(k)$ is precisely the set of edges of $G[C(k)]$ that do not lie in $C(k)$. This proves \ref{pr2}. 
If $\W_g^{(k-1)}$ does not contain a bad demand edge,
then $\W_g^{(k-1)}=\W_g^{(k)}$, as observed before. 
Otherwise, $\W_g^{(k-1)}$ contains precisely one bad demand edge $h$.
By \ref{pr3}, both end-points $s,t$ of $h$ lie on $C^{(k-1)}$. 
If neither $s$ nor $t$ are contained in $\{i_2,i_3,\hdots,i_\ell\}$, 
then $h$ is not contained in $\Dp_1\cup\Dp_2$ during iteration $k-1$;
in which case $\Wk_g=\Wkminus_g$ and $s,t\in G[C^{(k)}]$.
Otherwise, $h\in \Dp_1\cup\Dp_2$. 
If $h\in \Dp_2$, then, by the same argument as in Case 2,
(i) $\Wk_g=(\Wkminus_g-h)\cup \Dpminus_h$ only contains \emph{good} demand edges; and (ii) $\Wk_g$ is the edge-set of a walk between the end-points of $g$.
%$\W_g^{(k)}$ consists of edges whose endpoints are adjacent in $G$. 
Thus assume that $h\in \Dp_1$, $s\in \{i_2,i_3,\hdots,i_{\ell-1}\}$, and $t\in \{\omega_1,\omega_2,\hdots,\omega_p\}$. 
If $\Wk_g$ contains a bad demand edge,
then this edge is either $\{i_1,t\}$ or $\{i_\ell,t\}$;
both of these possible demands edges have their two end-points in $C^{(k)}$. 
Finally, since (a) $\Wkminus_g$ is the edge-set of a walk between the end-points of $g$, and (b) $h$ gets replaced in $\Wkminus_g$ with the edge-set $\Dpminus_h$ of an $s-t$ path (containing at most one bad demand edge),
then $\Wk_g$ is the edge-set of a walk between the end-points of $G$. This proves \ref{pr3}.
\end{proof}
%%%%%%%%%%%%%%%%%%%
%%%%%%%%%%%%%%%%%%%%

\begin{lemma}\label{lemma:instancesSatisfyCutCond}
Let $x$ be a feasible flow for $\I\left(G,\cp,\Dp,\dpr\right)$, then the following hold:
\noindent
\begin{enumerate}[label=(\alph*),topsep=0.2ex]
\item The ring-loading instance $\RL\left([n],\crl,\Dp_1\cup\Dp_2,\dpr\right)$ defined at \cref{ringLoadfeasible} satisfies the cut-condition.\label{condition2}
\item The instance $\mathcal{I}\left(G,x^*+\alpha \cdot \Dmax\cdot\bon_f,D_f,d^{(k+1)}\right)$ defined at \cref{cutconditionEndofIteration} satisfies the cut-condition.\label{condition3}
\end{enumerate}
\end{lemma}
\begin{proof}
Observe that at the start of iteration $k$, $x^*$ is a feasible flow of $\I(G,x^*,\Dp_1\cup\Dp_2,\dpr)$, and hence $\I(G,x^*,\Dp_1\cup\Dp_2,\dpr)$ satisfies the cut-condition, i.e.,
%(for all central cuts):
\begin{equation}\label{eq:cutconditionD1D2}
\dpr\left(\delta_{\Dp_1\cup\Dp_2}\left(\,\closedleft i,j\closedright\,\right)\right)\leq x^*\big(\delta_{G}\left(\,\closedleft i,j\closedright\,\right)\big) \text{ for each }i,j\in[n].
\end{equation}
Let $i,j\in[n]$ with $i\neq j+1$, and $f_1$ and $f_2$ be the unique faces in $\F-\fout$ that contain $\{i-1,i\}$ and $\{j,j+1\}$ respectively. Then,
\begin{align*}
\dpr\left(\delta_{\Dp_1\cup\Dp_2}\left(\,\closedleft i,j\closedright\,\right)\right)
&\leq  x^*\big(\delta_{G}\left(\,\closedleft i,j\closedright\,\right)\big)
\leq x^*(G^*_{f_1,f})+x^*(G^*_{f,f_2})+x^*(\{i-1,i\})+x^*(\{j,j+1\})\\
&=\crl\big(\,\{i-1,i\}\,\big)+\crl\big(\,\{j,j+1\}\,\big)=\crl\big(\,\delta_{[n]}(\,\closedleft i,j\closedright\,)\,\big).
\end{align*}
The first inequality follows from \eqref{eq:cutconditionD1D2}. The second inequality follows from Lemma~\ref{lemma:usefullemmafor[i,j]}. The first equality follows from the definition of $\crl$. 
%on line  \ref{definingEdgecapacitiesforRL}.
    %The last equation follows since the cut edges of $\closedleft i,j\closedright)$ are precisely $\{i-1,i\},\{j,j+1\}$. 
 It follows from Lemma~\ref{central} that $\RL([n],\crl,\Dp_1\cup\Dp_2,\dpr)$ satisfies the cut-condition. This proves \cref{condition2}. 
 Hence, Algorithm $\Alg$ computes a $\alpha$-feasible unsplittable flow $z$ for $\RL\left([n],\crl,\Dp_1\cup\Dp_2,\dpr\right)$. 
 We now prove \ref{condition3}. 
 By Lemma~\ref{central}, it suffices to show that,
% \david{Should we use some annoying notation like in here
    % https://arxiv.org/pdf/1606.07927 since there are parallel edges?}
\begin{equation}\label{eq:cutconditionendofalgorithm}
d^{(k+1)}\big(\delta_{D_f}\left(\,\closedleft i,j\closedright\,\right)\big)\leq (x^*+\alpha \cdot \Dmax\cdot\bon_f)(\delta_{G}\left(\,\closedleft i,j\closedright\,\right)\big) ~\text{for}~i,j\in[n], i\neq j+1.
\end{equation}
We consider three cases.

\noindent \textbf{Case 1: $i,j \in \openleft i_\ell,i_{1} \openright$}. 

It suffices to show that, 
$$\dpr\left(\delta_{\Dp_1\cup\Dp_2}\left(\,\closedleft i,j\closedright\,\right)\right)=d^{(k+1)}\big(\,\delta_{D_f}\left(\,\closedleft i,j\closedright\,\right)\big),$$
since inequality \eqref{eq:cutconditionendofalgorithm} would then immediately follow from \eqref{eq:cutconditionD1D2} in this case. Observe that: (i) all the edges in $\Dp_h$, for $h\in\Dp_2$, are only incident to nodes of $f$; and (ii) if $h\in \Dp_1$, then $t\in V(h)$ is the unique endpoint of $h$ contained in $\openleft i_\ell,i_1\openright$, and there is precisely one demand edge of $\Dp_h$ with exactly one endpoint in $\openleft i_\ell,i_1\openright$, and this endpoint is $t$. Then $\delta_{\Dp_2}\left(\,\closedleft i,j\closedright\,\right)=\emptyset$ and $\delta_{\Dp_h}\left(\,\closedleft i,j\closedright\,\right)=\emptyset$ for each $ h\in (\Dp_1\cup\Dp_2)\setminus\delta_{\Dp_1}\left(\,\closedleft i,j\closedright\,\right)$. This implies that:
\begin{align*}
\dpr\big(\delta_{\Dp_1\cup\Dp_2}(\closedleft i,j\closedright)\big)&=\sum_{h\in \delta_{\Dp_1}(\closedleft i,j\closedright)}\dpr(h);\\
\dpr\big(\delta_{D_f}(\closedleft i,j\closedright)\big)&=\sum_{h\in \delta_{\Dp_1}(\closedleft i,j\closedright)}\sum_{r\in \delta_{\Dp_h}(\closedleft i,j\closedright)}\dpr(r).
\end{align*}
Let $h\in \delta_{\Dp_1}(\closedleft i,j\closedright)$. By observation (ii) above, $|\delta_{\Dp_h}\left(\,\closedleft i,j\closedright\,\right)|=1$. Thus, since $\dpr(r)=\dpr(h)$ for each $r\in \Dp_h$, $\dpr(h)=\sum_{r\in \delta_{\Dp_h}(\closedleft i,j\closedright)}\dpr(r)$.
It follows that $\dpr\big(\delta_{\Dp_1\cup\Dp_2}\left(\,\closedleft i,j\closedright\,\right)\big)=\dpr\big(\delta_{D_f}\left(\,\closedleft i,j\closedright\,\right)\big)$, as desired.     

\noindent\textbf{Case 2:} $i,j \in \openleft i_q,i_{q+1} \openright$ for some $1\leq q\leq \ell-1$. 

Then either $\closedleft i,j\closedright\subseteq \openleft i_q,i_{q+1} \openright$ or $[n]\setminus\closedleft i,j\closedright\subseteq \openleft i_{q},i_{q+1} \openright$.
% %If $\closedleft i,j\closedright\subseteq \openleft i_q,i_{q+1} \openright$ ,
    % %then $\delta_{D_f}\left(\,\closedleft i,j\closedright\,\right)=\emptyset$ 
    Since for each $h\in \Dp_1\cup\Dp_2$, the demand edges in $\Dp_h$ 
    %reated on the for-loops on lines \ref{FirstForstarts} and \ref{SecondForStarts} 
    are incident only to nodes contained in either $f$ or $\openleft i_\ell,i_1 \openright$, 
    then either $[i,j]$ or $[n]\setminus[i,j]$ does not contain any node incident to edges of $D_f$.
    Hence $\delta_{D_f}\left(\,\closedleft i,j\closedright\,\right)=\emptyset$,
     and inequality \eqref{eq:cutconditionendofalgorithm} follows trivially.
    
\noindent\textbf{Case 3:} $\{i,j\} \nsubseteq \openleft i_q,i_{q+1} \openright$ for every $1\leq q\leq \ell$, where $i_{\ell+1}:=i_1$.

Then both $\closedleft i,j\closedright$ and $[n]\setminus [i,j]$ contain at least one node from $f$. Let $f_1$ and $f_2$ denote the unique faces in $\F-\fout$ that contain $\{i-1,i\}$ and $\{j,j+1\}$, respectively.
Since $z$ is an $\alpha$-feasible unsplittable flow for $\RL\left([n],\crl,\Dp_1\cup\Dp_2,\dpr\right)$,
and by construction of $D_f$,
the ring-loading instance $\RL\left([n],\crl+\alpha \cdot \Dmax,D_f,d^{(k+1)}\right)$ is feasible.
Therefore $\RL\left([n],\crl+\alpha \cdot \Dmax,D_f,d^{(k+1)}\right)$ satisfies the cut-condition. 
%Let $\big(\delta_{[n]}\left(\,\closedleft i,j\closedright\,\right)\big)$ denote the total capacity of edges going across the set $[i,j]$ in the ring-loading instance $\mathcal{RL}([n],\crl+\alpha \cdot \Dmax,D_f,\dpr)$. 
It follows that,
\begin{align*}
         d^{(k+1)}\big(\delta_{D_f}\left(\,\closedleft i,j\closedright\,\right)\big)
        &\leq  \left(\crl+\alpha \cdot \Dmax\right)\big(\delta_{[n]}\left(\,\closedleft i,j\closedright\,\right)\big)\\
         &= \crl\big(\,\{i-1,i\}\,\big)+\crl\big(\,\{j,j+1\}\,\big)+2 \cdot \alpha \cdot \Dmax\\
        &=x^*\left(G^*_{f,f_1}\right)+x^*\big(\,\{i-1,i\}\,\big)+x^*\left(G^*_{f,f_2}\right)+x^*\big(\,\{j,j+1\}\,\big)+2 \cdot \alpha \cdot \Dmax\\
        &=x^*\left(G^*_{f_1,f_2}\right)+x^*\big(\,\{i-1,i\}\,\big)+x^*\big(\,\{j,j+1\}\,\big)+2 \cdot \alpha \cdot \Dmax\\
        &=\left(x^*+\alpha \cdot \Dmax \cdot \bon_f\right)\big(\delta_{G}\left(\,\closedleft i,j\closedright\,\right)\big).
\end{align*}
 The second equality follows from the definition of $\crl$.
The third equality follows from Lemma~\ref{lemma:usefullemmafor[i,j]}.
The last equality follows from part (c) of Lemma~\ref{lemma:usefullemmafor[i,j]}, and the fact that there are precisely two edges of $f$ in $\delta_{G}\left(\,\closedleft i,j\closedright\,\right)$. This proves \cref{condition3}.
\end{proof} 
\begin{comment}
\textcolor{red}{
We sketch the proof of the critical Lemma~\ref{lemma:instancesSatisfyCutCond}, which makes a repeated use of Lemma~\ref{lemma:usefullemmafor[i,j]}.
The ring-loading instance $\RL([n],\crl,\Dp_1\cup\Dp_2,\dpr)$ satisfies the cut-condition (i.e., it is feasible) because (1) $\mathcal{I}(G,x^*,\Dp_1\cup\Dp_2,\dpr)$ satisfies the cut-condition, and (2) the capacity of any central cut is at least as large in the ring with respect to $\crl$ as it is in $G$ with respect to $x^*$. 
Because of this and by definition of $D_f$, the instance $\RL\left([n],\crl+\alpha\cdot\Dmax,D_f,d^{(k+1)}\right)$ satisfies the cut-condition.
From this and (1), it follows that the instance $\mathcal{I}(G,x^*+\alpha \cdot \Dmax\cdot\bon_f,D_f,d^{(k+1)})$ satisfies the cut-condition, by noting that any central cut induced by some $[i,j]$ for which $d^{(k+1)}\left(\delta_{D_f}([i,j])\right)>\dpr\left(\delta_{\Dp_1\cup\Dp_2}([i,j])\right)$, satisfies the crucial property that the capacity of the cut induced by $[i,j]$ is the same in both the ring with capacities given by $\crl+\alpha\cdot\Dmax$ and the graph $G$ with capacities $x^*+\alpha\cdot \Dmax \cdot\bon_f$. 
}
\end{comment}
A consequence of Lemma~\ref{lemma:instancesSatisfyCutCond} is that the instance $\I(G,\cp,\Dp,\dpr)$ computed by the algorithm at the end of an iteration is feasible.
\begin{lemma} \label{lemma:every_ite_feasible}
Let $x^{(k)}=\sum_{h\in\Dp}x_h$ be a feasible flow for $\I\left(G,\cp,\Dp,\dpr\right)$, then $x^{(k+1)}$ is a feasible flow for $\I\left(G, c^{(k+1)}, D^{(k+1)}, d^{(k+1)}\right)$. 
\end{lemma}
\begin{proof}
Suppose that $f\in\F-\fout$ is the face being considered in iteration $k$. Since $x^{(k)}$ is a feasible flow for instance $\I\left(G,\cp,\Dp,\dpr\right)$, it follows that 
%and $x^\prime=x-x^*$, where $x^*=\sum_{h\in \Dp_1\cup\Dp_2}x_h$, then
$x^{(k)}-x^*$ is a feasible flow for instance 
$\I\left(G,\cp-x^*,\Dp\setminus \left(\Dp_1\cup\Dp_2\right),\dpr\right)$.
By Lemma~\ref{lemma:instancesSatisfyCutCond}, $\mathcal{I}\left(G,x^*+\alpha \cdot \Dmax\cdot\bon_f,D_f,d^{(k+1)}\right)$ (defined at \cref{cutconditionEndofIteration}) satisfies the cut-condition. Therefore, by Theorem~\ref{OS} a feasible flow $y$ of $\mathcal{I}\left(G,x^*+\alpha \cdot \Dmax\cdot\bon_f,D_f,d^{(k+1)}\right)$ can be computed (in polynomial time). It then follows that $x^{(k)}-x^*+y$ is a feasible flow for $\I\left(G,\cp+\alpha \cdot \Dmax\cdot\bon_f,\left(\Dp\setminus \left(\Dp_1\cup\Dp_2\right)\right)\cup D_f,d^{(k+1)}\right)$
The proof of the lemma follows from the fact that $x^{(k+1)}= x^{(k)}-x^*+y$, and $$\I(G, c^{(k+1)}, D^{(k+1)}, 
 d^{(k+1)})=\I\left(G,\cp+\alpha \cdot \Dmax\cdot\bon_f,\left(\Dp\setminus \left(\Dp_1\cup\Dp_2\right)\right)\cup D_f,d^{(k+1)}\right).$$
\end{proof}

We are now ready to give a proof of Theorem~\ref{theorem:Newinstance}, which we restate.
\theoremnew*
\begin{proof}By Lemma~\ref{lemma:every_ite_feasible}, the cut-condition is satisfied for each $\I(G,\cp,\Dp,\dpr)$, and hence the cut-condition is also satisfied for the final output $\I(G,c',D',d')$. Note that the capacity of an edge $e$ is increased by $\alpha \cdot d_{\max}$ only when its label is also increased by 1. Initially, the label of all the edges on the unbounded face $\fout$ is 1 while the label of the remaining edges is 0. At the end of the algorithm, the label of every edge is 2. This implies that the capacity of edges on $\fout$ increases by $\alpha \cdot \Dmax$ and the capacity of the remaining edges increases by $2 \cdot \alpha \cdot \Dmax$.
By \cref{pr3} of Lemma~\ref{propositionCycle2}, it follows that the edges in $W_g$ form a walk between the end points of $g$. Furthermore, if there are bad edges in $W_g$, then both its end-points are incident on the cycle formed by label 1 edges. At the end of the algorithm, there are no label 1 edges and hence there are no bad edges in $W_g$. Whenever we insert a new demand edge in $\W_g^{(k)}$, we set its value to $d(g)$. Hence all the demand edges in $W_g$ have value $d(g)$, and this completes the proof.
%of Theorem~\ref{theorem:Newinstance}.
\end{proof}
\subsection{Handling Parallel Edges in $G$}\label{parallele}
Suppose that $G=(V,E)$ has parallel edges. We describe how to proceed in this case.
For every $u,v\in V$, let $E(u,v)$ denote the edges incident to both $u$ and $v$.
We can create a new graph $G_0=(V,E_0)$ by replacing each $E(u,v)\neq\emptyset$ with an edge $uv$ of capacity $c_0(uv)=\sum_{e\in E(u,v)}c(e)$.
Then $\I(G,c,D,d)$ is feasible if and only if $\I(G_0,c_0,D,d)$ is feasible. 
We can apply the algorithm in this section to $\I(G_0,c_0,D,d)$ in order to obtain an instance $\I(G_0,c_0^\prime,D^\prime,d^\prime)$ satisfying the properties in Theorem~\ref{theorem:Newinstance}.
For each $uv\in E_0$, if $c_0^\prime(uv)=c_0(uv)+\delta$, then pick an edge $e\in E(u,v)$ and define $c^\prime(e)=c(e)+\delta$ and $c^\prime(e^\prime)=c(e^\prime)$ for all $e^\prime\in E(u,v)-e$.
Then , the instance $\I(G,c^\prime,D^\prime,d^\prime)$ satisfies the properties of Theorem~\ref{theorem:Newinstance}  with respect to $\I(G,c,D,d)$.

\section{Proof of Theorem~\ref{theorem:G+Hplanar}}\label{section:good}
In this section, we give a proof of Theorem~\ref{theorem:G+Hplanar}, which we restate for convenience.
\theoremgood*
As discussed in Section 2, we may assume without loss of generality that $G=(V,E)$ is 2-vertex connected. 
%We assume that $G$ can have parallel edges. Since $G$ is outerplanar and the cut-condition is satisfied, there exists a feasible flow $x$ for $\I(G,c,H,d)$ by Theorem~\ref{OS}.
For any pair $u\neq v\in V$ of nodes, we use $E(u,v):=\{e\in E:V(e)=\{u,v\}\}$ and $D(u,v):=\{g\in D:V(g)=\{u,v\}\}$ to denote the set of edges of $E$ and $D$, respectively, incident to both $u$ and $v$.
We say that an unbounded face $f=i_1i_2\hdots i_\ell i_1$ of $G$ is an \emph{ear incident to $\{i_\ell, i_1\}$}, 
if $\ell\geq 3$ and if for every $j\in[\ell-1]$, there exists an edge $e\in E(i_\ell,i_{j+1})$ incident to the unbounded face of $G$.
\footnote{Observe that if we create a graph $G_0$ by replacing each $E(u,v)\neq\emptyset$ in $G$ with a single edge $uv$, then an ear of $G$ corresponds to a leaf of the weak dual of $G_0$.} 

We first give a brief description of the algorithm.
Our algorithm works in iterations; one per each unbounded face of $G$ incident to at least $3$ vertices.
At the start of iteration $k$ we consider a feasible instance $\I(G^{(k)},c^{(k)},D^{(k)},d^{(k)})$ where all demand edges are good and where $G^{(k)}$ is 2-vertex connected .
Initially, $\I(G^{(1)},c^{(1)},D^{(1)},d^{(1)})=\I(G,c,D,d)$.
We consider an ear $f$ of $G^{(k)}$ incident to some $\{i_\ell,i_1\}$. 
%All demands in $\cup_{j=1}^{\ell-1}D^{(k)}(i_j,i_{j+1})$ are either fully or partially routed on iteration $k$, using only edges in $\cup_{j=1}^{\ell-1}E^{(k)}(i_j,i_{j+1})$. 
For each $j\in[\ell-1]$, each demand $g\in D^{(k)}(i_j,i_{j+1})$ is either routed along a single edge of $E^{(k)}(i_j,i_{j+1})$, or $g$ is replaced with a demand of value $d^{(k)}(g)$ incident to $i_1,i_\ell$ after routing a flow of value $d^{(k)}(g)$ both from $i_j$ to $i_1$, and from $i_{j+1}$ to $i_\ell$, across two paths with node sequence $i_j\hdots i_1$ and $i_{j+1}\hdots i_\ell$, respectively.
We do this so that the capacity of any edge is not exceeded by more than $\Dmax$.
After this procedure, there are no demands incident to $i_2,\hdots,i_{\ell-1}$.
If there is one $j\in[\ell-1]$, 
such that all edges $e$ in $E^{(k)}(i_j,i_{j+1})$ have exceeded their original capacity $c^{(k)}(e)$, we remove $\{i_2,\hdots,i_{\ell-1}\}$ from $G^{(k)}$; otherwise, we replace $\{i_2,\hdots,i_{\ell-1}\}$ and their incident edges with a single edge $e_0=\{i_1,i_\ell\}$ with a suitable capacity.
We then continue to iteration $k+1$, and in there we consider any ear from the graph of the resulting instance.
%The resulting instance is then passed to iteration $k+1$.
%Note that the number of unbounded faces in the graph with more than two vertices decreases by one.
After the last iteration, the resulting (feasible) instance has a supply graph with only two vertices, which is straightforward to solve, for instance, by using Claim \ref{claim:pathgood} below. 

\noindent Iteration $k$ proceeds differently, depending on which of the following two cases hold:
\begin{align}\label{gooddemands:cases}
\textbf{Case 1:}& \sum_{g\in D^{(k)}(i_j,i_{j+1})}d^{(k)}(g)\leq \Dmax+\sum_{e\in E^{(k)}(i_j,i_{j+1})}c^{(k)}(e) \text{ for each }j\in [\ell-1].\nonumber\\
\textbf{Case 2:}& \sum_{g\in D^{(k)}(i_j,i_{j+1})}d^{(k)}(g)> \Dmax+\sum_{e\in E^{(k)}(i_j,i_{j+1})}c^{(k)}(e) \text{ for precisely one }j\in [\ell-1].
\end{align}
Observe that if $\I(G^{(k)},c^{(k)},D^{(k)},d^{(k)})$ satisfies the cut-condition, then it cannot be the case that there are $1\leq j< j^{\prime}\leq \ell-1$ with  $\sum_{g\in D^{(k)}(i_j,i_{j+1})}d^{(k)}(g)> \Dmax+\sum_{e\in E^{(k)}(i_j,i_{j+1})}c^{(k)}(e)$ and  $\sum_{g\in D^{(k)}(i_{j^\prime},i_{{j^\prime}+1})}d^{(k)}(g)>\Dmax+\sum_{e\in E^{(k)}(i_{j^\prime},i_{j^\prime+1})}c^{(k)}(e)$.
This holds because,
\begin{align*}
\sum_{e\in E(i_j,i_{j+1})}c^{(k)}(e)+\sum_{e\in E(i_{j^\prime},i_{j^\prime+1})}c^{(k)}(e)&=c^{(k)}(\{i_{j+1},i_{j+2},\hdots,i_{j^\prime}\})\geq d^{(k)}(\{i_{j+1},i_{j+2},\hdots,i_{j^\prime}\})\\
&=\sum_{g\in D^{(k)}(i_j,i_{j+1})}d^{(k)}(g)+\sum_{g\in D^{(k)}(i_{j^\prime},i_{j^\prime+1})}d^{(k)}(g),
\end{align*}
where the equalities follow from the fact that $f=i_1i_2\hdots i_\ell i_1$ is an ear of $G^{(k)}$.
%incident to $\{i_\ell,i_1\}$).

In order to prove Theorem~\ref{theorem:G+Hplanar}, we make repeated use of the following straightforward claims.
\begin{claim}\label{claim:pathgood}
Consider an instance $\mathcal{I}(G,c,D,d)$.
Let $v_1,v_2\hdots,v_r\in V$ be a set of distinct vertices, and let $D^*\subseteq D$ be a set of demand edges such that $V(g)=\{v_1,v_r\}$ for all $g\in D^*$.
If 
\begin{equation*}
\sum_{g\in D^*}d(g)\leq\Dmax+\sum_{e\in E(v_j,v_{j+1})}c(e),\;\forall j\in[r-1],
\end{equation*}
then one can efficiently compute a set of paths $\{p_g\}_{g\in D^*}$ in $G$ such that the node sequence of each $p_g$ is precisely $v_1v_2\hdots v_r$, and such that 
\begin{equation*}
\sum_{g\in D^*:e\in p_g}d(g)\leq c(e)+\Dmax,\;\forall e\in \{E(v_j,v_{j+1})\}_{j\in[r-1]}.
\end{equation*}
\end{claim}
\begin{proof}
%By the Max Flow-Min Cut Theorem,
Let $j\in[r-1]$.
Suppose that $E(i_j,i_{j+1})=\{e_1,\hdots,e_a\}$ and $D(i_j,i_{j+1})=\{g_1,\hdots,g_b\}$. 
%We prove by induction on $|E(i_j,i_{j+1})|$
It suffices to prove that we can compute an assignment $\sigma_j:D^*\rightarrow E(i_j,i_{j+1})$, such that $\sum_{g\in D^*:\sigma_j(g)=e}d(g)\leq c(e)+\Dmax$ for each $e\in E(v_j,v_{j+1})$. 
The proof of the theorem would then follow by taking $p_g=\{\sigma_j(g)\}_{j\in[\ell-1]}$.
We prove the above by induction on $|E(i_j,i_{j+1})|$.
Clearly, the above holds if $|E(i_j,i_{j+1})|=1$. 
Thus, suppose that $|E(i_j,i_{j+1})|>1$.
Take any edge $e^\prime\in E(v_j,v_{j+1})$, and let $Q\subseteq D^*$ be a maximal inclusion-wise subset of demands with $\sum_{g\in Q}d(g)\leq c(e^\prime)+\Dmax$.
Let $\sigma_j(g)=e^\prime$ for all $g\in Q$.
If $Q=D^*$, we are done. 
If $Q\neq D^*$, then $c(e^\prime)<\sum_{g\in Q}d(g)$.
Therefore, $\sum_{g\in D^*\setminus Q}d(g)\leq \Dmax+\sum_{e\in E(v_j,v_{j+1})-e^\prime}c(e)$. By the induction hypothesis, we can assign each $g\in D^*\setminus Q$ to some edge in $E(v_j,v_{j+1})-e^\prime$, so that $\sum_{g\in D^*\setminus Q:\sigma_j(g)=e}d(g)\leq c(e)+\Dmax$ for each $e\in E(v_j,v_{j+1})-e^\prime$. 
\end{proof}
\begin{claim}\label{claim:residual}
Suppose that an instance $\mathcal{I}(G,c,D,d)$ satisfies the cut-condition.
Let $g\in D$ be a good demand with $V(g)=\{u,v\}$. 
Suppose that one routes $g$ through some $e^\prime\in E(u,v)$, 
by possibly exceeding the capacity of $e^\prime$. 
Then the residual instance $\mathcal{I}(G,c^\prime,D-g,d)$ satisfies the cut-condition, where $c^\prime(e^\prime)=\max\{0,c(e^\prime)-d(g)\}$ and $c^\prime(e)=c(e)$ for each $e\in E-e^\prime$.
\end{claim}
\begin{proof}
Let $S\subseteq V$. 
Clearly, the cut-condition holds if $u,v\in S$ or $u,v\in V\setminus S$.
If suffices to consider the case in which, without loss of generality, $u\in S$ and $v\notin S$.
Observe that,
$$c^{\prime}(\delta_G(S))=c(\delta_G(S))-c(e^\prime)+c^\prime(e^\prime)\geq d(\delta_D(S))-c(e^\prime)+c^\prime(e^\prime)\geq d(\delta_D(S))-d(g).$$
\end{proof}
For each case in \eqref{gooddemands:cases}, we show how to construct the feasible instance $\I(G^{(k+1},c^{(k+1)},D^{(k+1)},d^{(k+1)})$, and show how to convert a 1-feasible unsplittable flow in $\I(G^{(k+1)},c^{(k+1)},D^{(k+1)},d^{(k+1)})$ into a 1-feasible unsplittable flow in $\I(G^{(k)},c^{(k)},D^{(k)},d^{(k)})$.

\subsection{Description of an Iteration in Case 1}
Suppose that $\sum_{g\in D^{(k)}(i_j,i_{j+1})}d^{(k)}(g)\leq \Dmax+\sum_{e\in E^{(k)}(i_1,i_{j+1})}c^{(k)}(e)$ for all $j\in[\ell-1]$.
By Claim~\ref{claim:pathgood}, we can assign each $g\in D^{(k)}(i_j,i_{j+1})$ to an edge $e_g\in E^{(k)}(i_j,i_{j+1})$, so that $d_e\leq c(e)+\Dmax$ for each $e\in E^{(k)}(i_j,i_{j+1})$,
where
$d_e:=\sum_{g\in D^{(k)}(i_j,i_{j+1}):e_g=e}d^{(k)}(g)$.
%For each $e\in E^{(k)}(i_j,i_{j+1})$,
%let $d_e=\sum_{g\in D^{(k)}(i_j,i_{j+1}):e_g=e}d^{(k)}(g)$.
After this operation, define $E_j=\{e\in E^{(k)}(i_j,i_{j+1}):d_e\leq c^{(k)}(e)\}$ as the set of edges whose capacity was not exceeded during the above procedure.
Create a new edge $e_0$ incident to $i_1$ and $i_\ell$, and define 
\begin{align*}
V^{(k+1)}&:=V^{(k)}\setminus\{i_2,\hdots,i_{\ell-1}\},\\
E^{(k+1)}&:=\left(E^{(k)}\setminus \cup_{j\in [\ell-1]}E^{(k)}(i_j,i_{j+1})\right)+e_0,\\
c^{(k+1)}(e)&:=\begin{cases}
c^{(k)}(e),&\:\text{if }~e\in E^{(k)}\setminus \cup_{j\in [\ell-1]}E^{(k)}(i_j,i_{j+1})\\
\min_{j\in[\ell-1]}\sum_{e\in E_j}(c^{(k)}(e)-d_e),&\:\text{if }~e=e_0,
\end{cases}\\
D^{(k+1)}&:=D^{(k)}\setminus \cup_{j\in [\ell-1]}D^{(k)}(i_j,i_{j+1}),\\
d^{(k+1)}(g)&:=d^{(k)}(g).
\end{align*}
\begin{claim}
$\I(G^{(k+1)},c^{(k+1)},D^{(k+1)},d^{(k+1)})$ is feasible.
\end{claim}
\begin{proof}
By Claim~\ref{claim:residual}, 
the residual instance $\I_{res}$ obtained from $\I(G^{(k)},c^{(k)},D^{(k)},d^{(k)})$ after routing the demands in $\{D^{(k)}(i_j,i_{j+1})\}_{j\in[\ell-1]}$,
replacing $\{E^{(k)}(i_j,i_{j+1})\}_{j\in[\ell-1]}$ with $\{E_j\}_{j\in[\ell-1]}$, and decreasing $c^{(k)}(e)$ by $d_e$, is feasible.
The claim follows from the Max Flow-Min Cut Theorem, by
observing that $\I_{res}$ does not have demands incident to any node in $\{i_2,\hdots,i_{\ell-1}\}$, and that $\I(G^{(k+1)},c^{(k+1)},D^{(k+1)},d^{(k+1)})$ is obtained from $\I_{res}$ by replacing $\{E_j\}_{j\in[\ell-1]}$ with $e_0$.
\end{proof}
\begin{claim}\label{flowk:case1}
A 1-feasible unsplittable flow $y$ in $\I(G^{(k+1)},c^{(k+1)},D^{(k+1)},d^{(k+1)})$ can be efficiently converted into a 1-feasible unsplittable flow $z$ in $\I(G^{k},c^{(k)},D^{(k)},d^{(k)})$.
\end{claim}
\begin{proof}
Let $y(g)$ be the edge set of the path on which $y$ routes $g\in D^{(k+1)}$.
For each $g\in \cup_{j\in[\ell-1]} D^{(k)}(i_j,i_{j+1})$, let $z(g)=\{e_g\}$.
For the remaining demands of $D^{(k)}$ i.e., the demands in $D^{(k+1)}$, we consider two cases.
For each $g\in D^{(k+1)}$ with $e_0\notin y(g)$, let $z(g)=y(g)$.
Let $D_0\subseteq D^{(k+1)}$ be the set of demands $g$ with $e_0\in y(g)$.
Let $y(e_0)$ be the flow that traverses $e_0$.
Then,
\begin{equation*}
y(e_0)\leq \Dmax+ c^{(k+1)}(e_0)=\Dmax+\min_{j\in[\ell-1]} \sum_{e\in E_j}(c^{(k)}(e)-d_e).
\end{equation*}
It follows from Claim~\ref{claim:pathgood} that we can compute a set of paths $\{p_g\}_{g\in D_0}$ in $\cup_{j\in[\ell-1]}E_j$, each of which has the node sequence $i_1i_2\hdots i_\ell$, such that for each $e\in \cup_{j\in[\ell-1]}E_j$,
\begin{equation*}
\sum_{g\in D_0:e\in p_g}d^{(k)}(g)\leq c^{(k)}(e)-d_e+\Dmax.
\end{equation*}
For each $g\in D_0$, we set $z(g)=(y(g)-e_0)\cup p_g$.
Therefore, for each $e\in \cup_{j\in[\ell-1]} E_j$,
\begin{align*}
z(e)&=\sum_{g\in D^{k}(i_j,i_{j+1}):z(g)=e}d^{(k)}(g)+\sum_{g\in D_0:e\in p_g}d^{(k)}(g)\\
%\sum_{g\in D^{(k)}:e\in z(g)}d^{(k)}(g)&=\sum_{g\in D^{k}(i_j,i_{j+1}):z(g)=e}d^{(k)}(g)+\sum_{g\in D_0:e\in p_g}d^{(k)}(g)\\
&\leq d_e+c^{(k)}(e)-d_e+\Dmax\leq c^{(k)}(e)+\Dmax.
\end{align*}
%Finally, note that none of the edges outside $\cup_{j\in[\ell-1]}E_j$ is used to route the demands in $D_0$.
Finally, note that $z(e)=d_e\leq c^{(k)}(e)+\Dmax$ for each $e\in \cup_{j\in[\ell-1]}\left(E^{(k)}(i_j,i_{j+1})\setminus E_j\right)$, 
and that $z(e)=y(e)\leq c^{(k+1)}(e)+\Dmax=c^{(k)}(e)+\Dmax$ for each $e\in E^{(k)}\setminus\cup_{j\in[\ell-1]}E^{(k)}(i_j,i_{j+1})$.
\begin{comment}
$e\in E^{(k)}\setminus\cup_{j\in[\ell-1]}E^{(k)}(i_j,i_{j+1})$, 
$z(e)=y(e)\leq c^{(k+1)}(e)+\Dmax=c^{(k)}(e)+\Dmax,$
$$\sum_{g\in D^{(k)}:e\in z(g)}d^{(k)}(g)=y(e)\leq c^{(k+1)}(e)+\Dmax=c^{(k)}(e)+\Dmax,$$
and for each $e\in \cup_{j\in[\ell-1]}\left(E^{(k)}(i_j,i_{j+1})\setminus E_j\right)$, $$\sum_{g\in D^{(k)}:e\in z(g)}d^{(k)}(g)=d_e\leq c^{(k)}(e)+\Dmax.$$
\end{comment}
\end{proof}

\subsection{Description of an Iteration in Case 2}
Let $j^\prime$ be the unique index in $[\ell-1]$ for which $\sum_{g\in D^{(k)}(i_{j^\prime},i_{j^\prime+1})}d^{(k)}(g)> \Dmax+\sum_{e\in E^{(k)}(i_{j^\prime},i_{j^\prime+1})}c^{(k)}(e)$.
As in Case 1, for each $j^\prime\neq j\in[\ell-1]$,
we can assign each $g\in D^{(k)}(i_j,i_{j+1})$ to an edge $e_g\in E^{(k)}(i_j,i_{j+1})$, 
so that $d_e\leq c(e)+\Dmax$ for each $e\in E^{(k)}(i_j,i_{j+1})$, 
where $d_e=\sum_{g\in D^{(k)}(i_j,i_{j+1}):e_g=e}d^{(k)}(g)$.
Take an inclusion-wise maximal subset $D^*_{j^\prime}\subseteq D^{(k)}(i_{j^\prime},i_{j^\prime+1})$, 
such that each $g\in D^*$ can be mapped to an edge $e_g\in E^{(k)}(i_{j^\prime},i_{j^\prime+1})$, so that $\sum_{g\in D^*_{j^\prime}:e_g=e}d^{(k)}(g)\leq c(e)+\Dmax$ for each $e\in E^{(k)}(i_{j^\prime},i_{j^\prime+1})$.
If $D^*_{j^\prime}= D^{(k)}(i_{j^\prime},i_{j^\prime+1})$, we can proceed in the exact same way as in Case 1 and define $\I(G^{(k+1)},c^{(k+1)},D^{(k+1)},d^{(k+1)})$ as there.
Thus, assume that $D^*_{j^\prime}\neq D^{(k)}(i_{j^\prime},i_{j^\prime+1})$.
%By Claim~\ref{claim:residual}, the residual instance 
Let $D^\#_{j^\prime}=D^{(k)}(i_{j^\prime},i_{j^\prime+1})\setminus D^*_{j^\prime}$.
For each $g\in D^\#_{j^\prime}$, we create a new demand edge $g^\#$ incident to $i_1,i_\ell$.
Define,
\begin{align*}
 V^{(k+1)}&:=V^{(k)}\setminus \{i_2,\hdots,i_{\ell-1}\},\\ E^{(k+1)}&:=E^{(k)}\setminus \{E(i_j,i_{j+1})\}_{j\in [\ell-1]},\\
 c^{(k+1)}(e)&:=c^{(k)}(e), \\
D^{(k+1)}&:=\left( D^{(k)}\setminus \cup_{j\in[\ell-1]}D^{(k)}(i_j,i_{j+1})\right)\cup\{g^\#\}_{g\in D_{j^\prime}^\#}. \\
d^{(k+1)}(g)&:=\begin{cases}
d^{(k)}(g),&\:\text{if }~g\in D^{(k)}\setminus \cup_{j\in [\ell-1]}D^{(k)}(i_j,i_{j+1})\\
d^{(k)}(h),&\:\text{if }~g=h^\#.
\end{cases}   
\end{align*}
%We define $V^{(k+1)}=V^{(k)}\setminus \{i_2,\hdots,i_{\ell-1}\}$, $E^{(k+1)}=E^{(k)}\setminus \{E(i_j,i_{j+1})\}_{j\in [\ell-1]}$, 
%$D^{(k+1)}=\left( D^{(k)}\setminus \{D^{(k)}(i_j,i_{j+1})\}_{j\in[\ell-1]}\right)\cup\{g^\#\}_{g\in D_{j^\prime}^\#}$. 
%For each $e\in E^{(k+1)}$ define $c^{(k+1)}(e)=c^{(k)}(e)$. 
%For each $g\in D^\#_{j^\prime}$ define $d^{(k+1)}(g^\#)=d^{(k)}(g)$, and let $d^{(k+1)}(g)=d^{(k)}(g)$ for $g\in  D^{(k)}\setminus \{D^{(k)}(i_j,i_{j+1})\}_{j\in[\ell-1]}$.
\begin{claim}\label{claim:iaib}
$\I(G^{(k+1)},c^{(k+1)},D^{(k+1)},d^{(k+1)})$ is feasible.
\end{claim}
\begin{proof}
As in Case 1, define $E_j=\{e\in E^{(k)}(i_j,i_{j+1}):d_e\leq c^{(k)}(e)\}$ for each $j\in[\ell-1]$.
By Claim~\ref{claim:residual}, 
the residual instance $\I_{\text{res}}=\I(G^\prime=(V^{(k)},E^\prime),c^\prime,D^\prime,d^{(k)})$ obtained from $\I(G^{(k)},c^{(k)},D^{(k)},d^{(k)})$ after routing the demands in $\left(\cup_{j\in[\ell-1]}D^{(k)}(i_j,i_{j+1})\right)\setminus D^\#_{j^\prime}$ and deleting $\cup_{j\in[\ell-1]}(E^{(k)}(i_j,i_{j+1})\setminus E_j)$, is feasible.
%Note that $c^\prime(e)=c^{(k)}(e)-d_e\geq0$ for each $e\in\cup_{j\in[\ell-1]}E_j$, and $c^\prime(e)=c^{(k)}(e)$ for all other $e\in E^\prime$.
Since $E_{j^\prime}=\emptyset$, both $i_1$ and $i_{\ell}$ are cut-vertices of $G^\prime$. 
Therefore, the cut-condition is maintained if we replace (pin) each $g\in D^\#_{j^\prime}$ in $\I_{res}$ with the three demand edges $g_a=\{i_{j^\prime},i_1\},\,g^\#=\{i_1,i_{\ell}\},\,g_b=\{i_\ell,i_{j^\prime+1}\}$ of the same demand value as $g$.
We can then split $\I_{res}$ into the three feasible instances
$\I(G^{(k+1)},c^{(k+1)},D^{(k+1)},d^{(k+1)})$,
$\I_a=\I(G_a=(V_a,E_a),c^\prime,\{g_a:g\in D^\#_{j^\prime}\},d^\prime)$, and $\I_b=\I(G_b=(V_b,E_b),c^\prime,\{g_b:g\in D^\#_{j^\prime}\},d^\prime)$, where
$V_a=\{i_1,\hdots,i_{j^\prime}\}$, $V_b=\{i_{j^\prime+1},\hdots,i_\ell\}$, $E_a=\cup_{j\in[j^\prime-1]}E_j$, $E_b=\cup_{j\in\{j^\prime+1,\hdots,\ell-1\}}E_j$, and $d^\prime(g_a)=d^\prime(g_b)=d^{(k)}(g)$.
\end{proof}
\begin{claim}\label{flowk:case2}
A 1-feasible unsplittable flow $y$ in $\I(G^{(k+1)},c^{(k+1)},D^{(k+1)},d^{(k+1)})$ can be efficiently converted into a 1-feasible unsplittable flow $z$ in $\I(G^{k},c^{(k)},D^{(k)},d^{(k)})$.
\end{claim}
\begin{proof}
Let $z(g)=y(g)$ for each $g\in D^{(k)}\setminus\cup_{j\in[\ell-1]}D^{(k)}(i_j,i_{j+1})$, and let $z(g)=e_g$ for each $g\in \left(\cup_{j\in[\ell-1]}D^{(k)}(i_j,i_{j+1})\right)\setminus D^\#_{j^\prime}$. Let $\I_a$ and $\I_b$ be as in the proof of Claim \ref{claim:iaib}.
By Claim \ref{claim:pathgood}, we can compute two $1$-feasible unsplittable flows $z_a$, $z_b$ for $\I_a$ and $\I_b$, respectively.
For each $g\in D^\#_{j^\prime}$, we set $z(g)=z_a(g_a)\cup y(g^\#)\cup z_b(g_b)$.
Observe that $z(e)= d_e\leq c^{(k)}(e)+\Dmax$ for each $e\in \cup_{j\in[\ell-1]}\left (E^{(k)}(i_j,i_{j+1})\setminus E_j\right)$, and $z(e)=y(e)\leq c^{(k+1)}(e)+\Dmax=c^{(k)}(e)+\Dmax$ for each $e\in E^{(k)}\setminus \cup_{j\in[\ell-1]}E^{(k)}(i_j,i_{j+1})$. Finally, let $c\in\{a,b\}$ and $e\in E_c$. Then, $$z(e)=d_e+z_c(e)\leq d_e+c^\prime(e)+\Dmax=d_e+c^{(k)}(e)-d_e+\Dmax=c^{(k)}(e)+\Dmax.$$
\end{proof}
\subsection{Proof of Theorem~\ref{theorem:G+Hplanar}}
Let $\I(G,c,D,d)$ be a feasible instance and let $F$ be the number of unbounded faces of $G$ with at least 3 vertices.
For each iteration $k\in[F+1]$, we compute the feasible instances $\I(G^{(k)},c^{(k)},D^{(k)},d^{(k)})$ as described above. 
Observe that each instance only contains good demands, and that each $G^{(k)}$ is a 2-vertex connected planar graph. 
In the end, $\I(G^{(F+1)},c^{(F+1)},D^{(F+1)},d^{(F+1)})$ has no unbounded faces with at least 3 vertices. 
Thus, $G^{(F+1)}$ contains only 2 vertices.
By Claim~\ref{claim:pathgood}, we can compute a 1-feasible unsplittable flow in $\I(G^{(F+1)},c^{(F+1)},D^{(F+1)},d^{(F+1)})$.
By Claims~\ref{flowk:case1} and \ref{flowk:case2}, we can backtrack starting from $k=F$ in order to compute a 1-feasible unsplittable flow in $\I(G^{(1)},c^{(1)},D^{(1)},d^{(1)})=\I(G,c,D,d)$.
This concludes the proof.

\section{Conclusion}
We proved that if $G$ is an outerplanar graph and the cut-condition is satisfied, then there exists a $\mathcal{O}(1)$-feasible unsplittable flow. It is an interesting open problem to generalize this result to other settings where the cut-condition is also sufficient for a feasible routing of all the demands, for example when all the source-sink pairs are incident on the unbounded face of a planar graph. It would also be interesting to see if a similar result holds for a more general class of graphs, where the cut-condition is not sufficient for a feasible routing, such as series-parallel graphs.

\printbibliography

\end{document}